\documentclass[a4paper,UKenglish,cleveref, autoref, thm-restate]{lipics-v2021}

\usepackage{amssymb}
\usepackage{color}
\usepackage{nicefrac}
\usepackage{graphicx}
\usepackage{amsthm}
\usepackage{amsmath}

\newcommand{\cO}{\mathcal{O}}
\newcommand{\B}{\mbox{\bf B}}
\newcommand{\C}{\mbox{\bf C}}
\renewcommand{\L}{\mbox{\bf L}}
\newcommand{\DL}{\mathcal{L}}
\newcommand{\id}{\mathsf{id}}
\newcommand{\xorindex}{\texttt{XOR-index}}

\newcommand{\local}{\mbox{\sf LOCAL}}
\newcommand{\congest}{\mbox{\sf CONGEST}}
\newcommand{\bcc}{\mbox{\sf BCC}}
\newcommand{\ncc}{\mbox{\sf NCC}}
\newcommand{\ucc}{\mbox{\sf UCC}}

\title{Computing Power of Hybrid Models in Synchronous Networks} 

\author{Pierre Fraigniaud\thanks{Additional support for ANR projects QuData and DUCAT.}}{IRIF, Universit\'e Paris Cit\'e and CNRS, France.}{pierre.fraigniaud@irif.fr}{}{}
\author{Pedro Montealegre\thanks{This work was supported by Centro de Modelamiento Matem\'atico (CMM), ACE210010 and FB210005, BASAL funds for centers of excellence from ANID-Chile, FONDECYT 1220142, FONDECYT 11190482, and PAI77170068 }}{Facultad de Ingenier\'ia y Ciencias, Universidad Adolfo Ib\'a\~nez, Santiago, Chile.}{p.montealegre@uai.cl}{}{}
\author{Pablo Paredes}{Departamento de Ingenier\'ia Matem\'atica, Universidad de Chile, Chile}{pparedes@dim.uchile.cl}{}{}
\author{Ivan Rapaport}{DIM-CMM (UMI 2807 CNRS), Universidad de Chile, Chile.}{rapaport@dim.uchile.cl}{}{}
\author{Mart\'in R\'ios-Wilson}{Facultad de Ingenier\'ia y Ciencias, Universidad Adolfo Ib\'a\~nez, Santiago, Chile.}{martin.rios@uai.cl}{}{}
\author{Ioan Todinca}{LIFO, Universit\'e d'Orl\'eans and INSA Centre-Val de Loire, France.}{ioan.todinca@univ-orleans.fr}{}{}

\authorrunning{P. Fraigniaud et al.} 

\Copyright{Pierre Fraigniaud, Pedro Montealegre, Pablo Paredes, Mart\'in R\'ios-Wilson, Ivan Rapaport and Ioan Todinca} 


\category{} 

\nolinenumbers 

\EventEditors{John Q. Open and Joan R. Access}
\EventNoEds{2}
\EventLongTitle{42nd Conference on Very Important Topics (CVIT 2016)}
\EventShortTitle{CVIT 2016}
\EventAcronym{CVIT}
\EventYear{2016}
\EventDate{December 24--27, 2016}
\EventLocation{Little Whinging, United Kingdom}
\EventLogo{}
\SeriesVolume{42}
\ArticleNo{23}

\begin{document}

\maketitle

\begin{abstract}
During the last two decades, a small set of distributed computing models for networks have emerged, among which \local, \congest, and Broadcast Congested Clique (\bcc) play a prominent role. We consider \emph{hybrid} models resulting from combining these three models. That is, we analyze the computing power of models allowing to, say, perform a constant number of rounds of \congest, then  a constant number of rounds of  \local, then  a constant number of rounds of \bcc, possibly repeating this figure a constant number of times. We specifically focus on 2-round models, and we establish the complete picture of the relative powers of these models. That is, for every pair of such models, we determine whether one is (strictly) stronger than the other, or whether the two models are incomparable. The separation results are obtained by approaching communication complexity through an original angle, which may be of an independent interest.  The two players are not bounded to compute the value of a binary function, but the \emph{combined} outputs of the two players are constrained by this value. In particular, we introduce the {\tt XOR-Index} problem, in which Alice is given a binary vector $x\in\{0,1\}^n$ together with an index $i\in[n]$, Bob is given a binary vector $y\in\{0,1\}^n$ together with an index $j\in[n]$, and, after a single round of 2-way communication, Alice must output a boolean~$\mbox{out}_A$, and Bob must output a boolean~$\mbox{out}_B$, such that $\mbox{out}_A\land\mbox{out}_B = x_j\oplus y_i$. We show that the communication complexity of  {\tt XOR-Index} is $\Omega(n)$ bits. 
\end{abstract}

\clearpage
\section{Introduction}

This paper analyzes the relative power of  distributed computing models for networks, all resulting from the combination of standard synchronous models such as \local\/ and \congest~\cite{P00}, as well as Broadcast Congested Clique (\bcc)~\cite{drucker2014power}. Each of these three models has its strengths  and limitations. In particular, \congest\/ assumes the ability for each node to send a specific message to each of its neighbors at every round (even in a clique). However, the communication links have limited bandwidth. Specifically, at most $O(\log n)$ bits can be sent through any link during a round, in $n$-node networks. \local\/ assumes a link with unlimited bandwidth between any two neighboring nodes, but the information acquired by any node~$u$ after $t\geq 0$ rounds of communication is limited to the data available at nodes at distance at most~$t$ from~$u$ in the network. Finally, \bcc\/ supports all-to-all communications between the nodes, and thus does not suffer from the locality constraint of \local\/ and  \congest\/. However, at each round, each node is bounded to send a \emph{same} $O(\log n)$-bit message to all the other nodes. In this paper, we investigate the power of models resulting from combining these three models, in order to take advantage of their positive aspects without suffering from their negative ones. 

For the sake of comparing models, we focus on the standard framework of distributed \emph{decision} problems on labeled graphs (see~\cite{FeuilloleyF16}). Such problems are defined by a collection $\DL$ of pairs $(G,\ell)$, where $G=(V,E)$ is a graph, and $\ell:V\to\{0,1\}^*$ is a function assigning a label $\ell(u)\in\{0,1\}^*$ to every $u\in V$. Such a set $\DL$ is called a distributed \emph{language}. For instance, deciding whether a certain set~$U$ of nodes in a graph~$G$ forms a vertex cover can be modeled by the language 
\[
\mbox{\tt vertex-cover}=\big\{(G,\ell): \forall \{u,v\}\in E(G), \;\ell(u)=1 \lor \ell(v)=1\big\}, 
\] 
by labeling~1 all the vertices in~$U$, and~0 all the other vertices. Similarly, deciding $C_4$-freeness can be modeled by the language $\mbox{\tt $C_4$-freeness}=\{(G,\ell): C_4 \not\preceq G\}$, where $H\preceq G$ denotes that $H$ is a subgraph of~$G$, and deciding whether a graph is planar can be captured by the language $\mbox{\tt planarity}=\{(G,\ell): \mbox{$G$ is planar}\}$.  A distributed algorithm~$A$ \emph{decides}~$\DL$ if every node running~$A$ eventually accepts or rejects, and the following condition is satisfied: for every labeled graph $(G,\ell)$, 
\[
(G,\ell)\in \DL \iff \mbox{all nodes accept.}
\]
That is, every node should accept in a yes-instance (i.e., an instance $(G,\ell)\in\DL$), and, in a no-instance (i.e.,  an instance $(G,\ell)\notin\DL$), at least one node must reject. 
 
For every $t\geq 0$,  let us denote by $\L^t$ the set of distributed languages~$\DL$ for which there is a $t$-round algorithm in the \local\/ model deciding~$\DL$, with $\L = \L^1$. The sets $\C^t$ and $\B^t$ are defined similarly, for the \congest\/ and \bcc\/ models, respectively. Note that while it is easy to show, using indistinguishability arguments, that, for every $t\geq 1$, $\L^t\smallsetminus \L^{t-1}\neq\varnothing$ and $\C^t\smallsetminus \C^{t-1}\neq\varnothing$, establishing that there is indeed a decision problem in $\B^t\smallsetminus \B^{t-1}$ requires significantly more work~\cite{nisan1991rounds}. Also, we define $\L^*=\cup_{t\geq 0}\L^t$, $\C^*=\cup_{t\geq 0}\C^t$, and $\B^*=\cup_{t\geq 0}\B^t$. So, in particular, $\L^*$ is the class of distributed languages that can be decided in a constant number of rounds in the \local\/ model. 

The three models under consideration, i.e., \local, \congest, and \bcc\/ exhibit very different behaviors with respect to decision problems. For instance, it is known~\cite{DBLP:conf/podc/DruckerKO13} that 
\[
\mbox{\tt $C_4$-freeness}\in \L\smallsetminus (\B^*\cup\C^*), 
\] 
whenever one assumes, as we do in this paper, that, for all models under consideration, every node is initially aware of the identifiers\footnote{In each of the models, every node~$u$ of a $n$-node network $G=(V,E)$  is supposed to be provided with an identifier $\id(u)$, where $\id:V\to[1,N]$ is one-to-one, and $N(n)=\mbox{poly}(n)$, i.e., all identifiers can be stored on $O(\log n)$ bits in $n$-node networks. We also assume that all nodes are initially aware of the size~$n$ of the network, merely because this is the case in model \bcc.} of its neighbors. 
 On the other hand,  it is also known~\cite{BeckerKMNRST15} that 
\[
\mbox{\tt planarity}\in \B\smallsetminus \L^*.
\] 
This means that while no \local\/ algorithms can decide planarity in a constant number of rounds, there is a 1-round \bcc\/ algorithm deciding planarity, and while no \bcc\/ algorithms can decide $C_4$-freeness in a constant number of rounds, there is a 1-round \local\/ algorithm deciding $C_4$-freeness. So, if one allows \local\/  algorithms to do just a single round of all-to-all communication, as in \bcc\/, then both $C_4$-freeness and planarity can be solved in a constant number of rounds, hence increasing the computational power of  \local\/ dramatically. 

This observation led us to investigate scenarios such as the case in which the \congest\/ model is enhanced by allowing nodes to perform few rounds in either \local\/, or \bcc\/. What would be the computing power of such a \emph{hybrid}  model? For answering this question, for a collection of non-negative integers $\alpha_1,\dots,\alpha_k$, $\beta_1,\dots,\beta_k$, and $\gamma_1,\dots,\gamma_k$, we define the set 
\[
\prod_{i=1}^k\L^{\alpha_i}\B^{\beta_i}\C^{\gamma_i}
\]
as the class of decision languages~$\DL$ which can be decided by an algorithm performing $\alpha_1\geq 0$ rounds of  \local\/, followed by $\beta_1\geq 0$ rounds of  \bcc\/, followed by $\gamma_1\geq 0$ rounds of  \congest, followed by $\alpha_2\geq 0$ rounds of  \local\/, etc., up to $\gamma_k\geq 0$ rounds of  \congest. For instance, we have 
\[
\{\mbox{{\tt planarity},\;{\tt $C_4$-freeness}}\}\subseteq \L\B \cap  \B\L.
\]
However, how do  $\L\B$ and $\B\L$ compare? And what about $\C\B$ vs.~$\B\C$, and $\L\C$ vs.~$\C\L$? These are the kinds of questions that we are studying in this paper. In the long-term perspective, this line of research is motivated by the following question. Let $\DL$ be a fixed distributed language, and let us assume that a round of \local\/ costs~$a$ (say, for acquiring high-throughput channels), that a round of \bcc\/ costs~$b$ (say, for benefiting of facilities supporting all-to-all communications), and that a round of \congest\/ costs~$c$. The goal is to minimize the total cost of an algorithm deciding~$\DL$ in a constant number of rounds, that is, to solve the following minimization problem: 
\begin{equation}
\min_{\prod_{i=1}^k\L^{\alpha_i}\B^{\beta_i}\C^{\gamma_i} \mbox{\small $\;\ni \DL$}} 
\left ( a\sum_{i=1}^k\alpha_i+b\sum_{i=1}^k\beta_i+c\sum_{i=1}^k\gamma_i \right).
\label{eq:min}
\end{equation}
Note that, for $a=b=c=1$, Eq.~\eqref{eq:min}  corresponds to minimizing the number of rounds for deciding~$\DL$ when using a combination of the communication facilities provided by \local, \congest, and \bcc. For instance, deciding whether a graph is $C_k$-free can be achieved in $\lfloor\frac{k}{2}\rfloor$ rounds in \local, that is, $\mbox{\tt  $C_k$-freeness}\in \L^{\lfloor k/2 \rfloor}$. Eq.~\eqref{eq:min} is asking whether deciding {\tt $C_k$-freeness} could be achieved at a lower cost by combining  \local, \congest, and \bcc.  For tackling Eq.~\eqref{eq:min}, we need a better understanding of the fundamental effects resulting from combining  these models. 

\subsection{Our Results}

On the negative side, we provide a series of separation results between 2-round hybrid models. In particular, we show that $\B\C$ and $\C\B$ are incomparable. That is, there are languages in $\B\C\smallsetminus\C\B$, and languages in $\C\B\smallsetminus\B\C$. In fact, we show stronger separation results, by establishing that $\B\C\smallsetminus\C^*\B\neq\varnothing$, and $\C\B\smallsetminus\B\L^*\neq\varnothing$. That is, in particular, there are languages that can be decided by a 2-round algorithm performing a single \bcc\/ round followed by one \congest\/ round, which cannot be decided by any algorithm performing $k$ \congest\/ rounds followed by a single \bcc\/ round, for any $k\geq 1$. 

On the positive side, we show that, for any non-negative integers $\alpha_1,\dots,\alpha_k,\beta_1,\dots,\beta_k$, 
\begin{equation}
\prod_{i=1}^k \L^{\alpha_i}\B^{\beta_i} \subseteq \L^{\sum_{i=1}^k\alpha_i}\B^{\sum_{i=1}^k\beta_i}.
\label{eq:BL-in-LB}
\end{equation}
That is, if a language $\DL$ can be decided by a $t$-round algorithm alternating \local\/ and \bcc\/ rounds, then $\DL$ can be decided by a $t$-round algorithm performing all its \local\/ rounds first, and then all its \bcc\/ rounds --- with the notations of Eq.~\eqref{eq:BL-in-LB},  $t=\sum_{i=1}^k(\alpha_i+\beta_i)$. So, in particular $\B\L\subseteq \L\B$. This inclusion is strict, since, as said before, $\C\B\smallsetminus\B\L^*\neq \varnothing$. In fact, this separation holds even if the number of  \local\/ rounds depends on the number of nodes~$n$ in the network, as long as the algorithm performs $o(n)$ \local\/ rounds after its \bcc\/ round. Another consequence of  Eq.~\eqref{eq:BL-in-LB} is that the largest class of languages among all the ones considered in this paper is $\L^*\B^*$, that is, languages that can be decided by algorithms performing  $k$ \local\/ rounds followed by $k'$ \bcc\/ rounds, for some $k\geq 0$ and $k'\geq 0$. Thus,  Eq.~\eqref{eq:min} should be studied for languages~$\DL\in\L^*\B^*$.  

Interestingly, our separation results hold even for randomized protocols, which can err with probability at most~$\epsilon \leq \nicefrac15$. That is, in particular, there is a language $\DL\in \C\B$ (i.e., that can be decided by a deterministic 2-round algorithm) which cannot be decided with error probability at most $\nicefrac15$ by any randomized algorithm performing one \bcc\/ round first, followed by $k$ \local\/ rounds, for any $k\geq 1$. All our results about 2-rounds hybrid models are summarized on Figure~\ref{fig:lattice}.

\begin{figure}[htb]
\centering
\includegraphics[width=7cm]{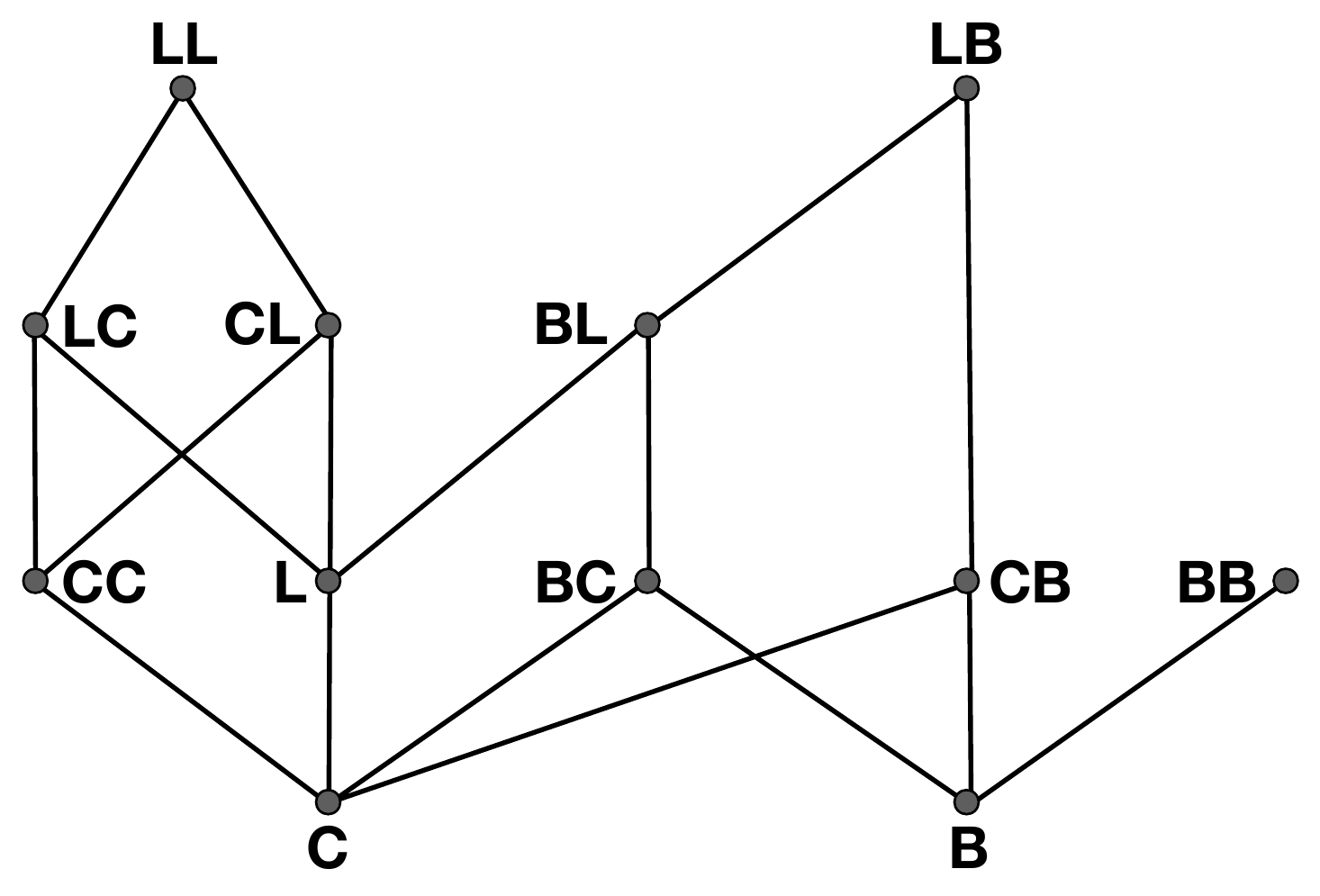}
\caption{\sl The poset of 2-round hybrid models. An edge between a set of languages $\mathbf{S}_1$ and a set~$\mathbf{S}_2$, where $\mathbf{S}_1$ is at a  level lower than $\mathbf{S}_2$, indicates that $\mathbf{S}_1\subseteq \mathbf{S}_2$. In fact, all inclusions are strict. Transitive edges are not displayed. Two sets that are not connected by a monotone path are incomparable. For instance, $\C\B$ and $\B\L$ are incomparable, while $\B\C\subseteq \L\B$. }
\label{fig:lattice}
\end{figure}

\subparagraph{Our Techniques.}

All our separation results are obtained by reductions from communication complexity lower bounds. However, we had to revisit several known communication complexity results for adapting them to the setting of distributed decision, in which  no-instances may be rejected by a single node, and not necessarily by all the nodes. In particular, we revisit the classical {\tt Index} problem. Recall that, in this problem, Alice is given a binary vector $x\in\{0,1\}^n$, Bob is given an index $i\in[n]$, and Bob must output $x_i$ based on a single message received from Alice (1-way communication). We define the {\tt XOR-Index} problem, in which Alice is given a binary vector $x\in\{0,1\}^n$ together with an index $i\in[n]$, Bob is given a binary vector $y\in\{0,1\}^n$ together with an index $j\in[n]$, and, after a single round of 2-way communication, Alice must output a boolean~$\mbox{out}_A$ and Bob must output a boolean~$\mbox{out}_B$, such that 
\[
\mbox{out}_A \land \mbox{out}_B = x_j\oplus y_i.
\]
That is, if $x_j\neq y_i$ then Alice and Bob must both accept (i.e., output \emph{true}), and if $x_j=y_i$ then \emph{at least} one of these two players must reject (i.e., output \emph{false}). We show that the sum of the sizes of the message sent by Alice to Bob and the message sent by Bob to Alice is $\Omega(n)$ bits. This bound holds even if the communication protocol is randomized and may err with probability at most $\nicefrac15$, and even if the two players have access to shared random coins. 

The fact that only one of the two players may reject a no-instance (i.e., an instance where $x_j\oplus y_i=0$), and not necessarily both, while a yes-instance must be accepted by both players, yields an asymmetry which complicates the analysis. We use information theoretic tools for establishing our lower bound. Specifically, we identify a way to decorrelate the behaviors of Alice and Bob, so that to analyze separately the distribution of decisions taken by each player, and then to recombine them for lower bounding the probability of error in case the messages exchanged between the players are small, contradicting the fact that this error probability is supposed to be small. Roughly, given  messages $m_A$ and $m_B$ exchanged by the two players, and given two indices $i$ and $j$, we compute the value $y_i$ maximizing the error probability for Alice, and the value $x_j$ maximizing the error probability for Bob, conditioned to $m_A,m_B,i,j$. We then show that the combined pair $(x_j,y_i)$ provide a sufficiently good lower bound on the probability of error for the whole protocol, which contradicts the fact that the error must be at most~$\epsilon$.

\subsection{Related Work}


The \local\/ model was introduced in~\cite{Linial92} at the beginning of the 1990s, when the celebrated $\Omega(\log^*n)$ lower bound on the number of rounds for computing a 3-coloring or a maximal independent set (MIS) in the $n$-node cycle was proved. A few years later, the class of \emph{locally checkable labeling} (LCL) problems was introduced and studied in~\cite{naor1995can}. This class essentially corresponds to the class~$\L^*$, but restricted to graphs with constant maximum degrees. Given $\DL\in\L^*$, and  the family $\mathcal{G}_\Delta$ of graphs with maximum degree at most~$\Delta$, solving the LCL problem induced by $\DL$ and $\mathcal{G}_\Delta$ consists of designing a distributed algorithm which, given a graph~$G\in\mathcal{G}_\Delta$, computes a labeling~$\ell$ of the nodes such that $(G,\ell)\in\DL$. It is known that many LCL problems can be solved in constant number of rounds in \local. This is for instance the case of certain types of weak colorings problems~\cite{naor1995can}. Also, there is an $O(\sqrt{k} \Delta^{1/\sqrt{k}} \log \Delta)$-approximation algorithm for  minimum dominating set running in $O(k)$ rounds~\cite{kuhn2004cannot} (where $k\geq 1$ is a parameter), and there is an $O(n^\varepsilon)$-approximation algorithm for the minimum coloring problem running in $\exp(O(1/\varepsilon))$ rounds~\cite{barenboim2018fast}. In fact, it is undecidable, in general, whether a given LCL problem has a (construction) algorithm running in a constant number of rounds~\cite{naor1995can}. A plethora of papers have addressed graph problems in the \local\/ model, and we refer to the survey~\cite{Suomela13}, but several significant results have been obtained since then, among which it is worth mentioning two fields in close connection to the topic of this paper, which emerged in the early 2010s. One is the systematic study of distributed decision problems in various settings, including non-determinism~\cite{FraigniaudKP13,GoosS16,KormanKP10} and interactive protocols~\cite{KolOS18,NaorPY20}. The other is a systematic study of the round-complexity of LCL problems (see, e.g., \cite{Balliu0OSST21,Suomela2020}, and the references therein). 


The \congest\/ model is a weaker variant of the \local\/ model in which the size of the messages exchanged at each round between neighbors is bounded to $O(\log n)$ bits, or $B$ bits in the parametrized version of the model.  This bound on the message size creates bottlenecks limiting the power of algorithms under this model. A fruitful line of research has established several non-trivial lower bounds on the round-complexity of \congest\/ algorithms, by reduction from communication complexity problems (see for instance
\cite{abboud2021smaller,artur2020detecting,elkin2006unconditional,peleg2000near,sarma2012distributed}). 
Nevertheless, several problems can still be solved in a constant number of rounds in \congest. This is for instance the case of computing a $(2+\varepsilon)$-approximation of minimum vertex cover which can be done in $O(\log \Delta / \log \log \Delta)$ rounds~\cite{bar2017distributed} in graphs with maximum degree~$\Delta$. Also, \emph{testing} (a weaker variant of decision, a la property-testing) the presence of specific subgraphs like small cliques or short cycles can be done in a constant number of rounds in~\congest (see, e.g., \cite{Censor-HillelFS19,EvenFFGLMMOORT17,FraigniaudO19,FraigniaudRST16,LeviMR21}).  


The congested clique model~\cite{drucker2014power,lotker2003mst} has first been introduced in its \emph{unicast} version (\ucc),
where every node is allowed to send potentially different $O(\log n)$-bit messages to each of the other $n-1$ nodes at every round.
In the \ucc\/ model, many natural problems can be solved in a constant number of rounds \cite{chang2019complexity,jurdzinski2018mst,lenzen2013optimal}. 
The \ucc\/ model is very powerful, and it has actually been proved~\cite{drucker2014power} that
it can simulate powerful bounded-depth circuits classes, from which it follows that exhibiting non-trivial lower bounds for the \ucc\/ model is quite difficult. The \emph{broadcast} variant of the congested clique, namely the  \bcc\/ model, is significantly weaker than the unicast variant, and lower bounds on the round-complexity of problems in the \bcc\/ model have been established, again by reduction to communication complexity problems. This is the case of problems such as detecting the presence of particular subgraphs~\cite{drucker2014power}, detecting planted cliques~\cite{chen2019broadcast}, or 
approximating the diameter of the network~\cite{holzer2016approximation}.
Obviously, many fast, non-trivial \bcc-algorithms  have also been devised. As examples, we can mention  
the sub-logarithmic deterministic algorithm that finds a maximal spanning forest in $O(\log n/\log \log n)$ rounds \cite{jurdzinski2018connectivity}, and algorithms for deciding and reconstructing several graph families (including bounded degeneracy graphs) performing in a constant number of rounds~\cite{becker2011adding}. It is worth noticing that, for single round algorithms, the \bcc\/ model is also referred to using other terminologies, such as \emph{simultaneous-messages}~\cite{babai2003communication}, or \emph{sketches}~\cite{ahn2012analyzing,yu2021tight}. 
In these latter models though, the measure of complexity is the size of the messages, and therefore the restriction to $O(\log n)$-bits   messages is not enforced.


Hybrid distributed computing models have  been  investigated in the literature only recently, motivated by the various forms of modern communication technologies, from high-throughput optical links to global wireless communication facilities, to peer-to-peer long-distance logical connections. In particular, a hybrid model allowing nodes to perform in a \emph{local} mode, and in a \emph{global} mode at each round has been recently considered~\cite{AugustineHKSS20}. The local mode corresponds to perform a \local\/ round~\cite{P00}, while the global mode corresponds to perform a \emph{node-capacitated clique} (\ncc) round~\cite{AugustineGGHSKL19}, which allows each node to exchange $O(\log n)$-bit messages with $O(\log n)$ arbitrary nodes in the network. It is shown that, in the \local+\ncc\/ hybrid model, SSSP can be approximated in $\widetilde{O}(n^{\nicefrac13})$ rounds, and APSP can be approximated in $\widetilde{O}(\sqrt{n})$ rounds. Several lower bounds are also presented in~\cite{AugustineHKSS20}, including an $\widetilde{\Omega}(\sqrt{n})$-round lower bound for computing APSP, and an $\Omega(n^{\nicefrac13})$-round lower bound for computing the diameter. In a subsequent work~\cite{KuhnS20}, it was shown that APSP can actually be solved exactly in $\widetilde{O}(\sqrt{n})$ rounds in the \local+\textsf{NCC} model. Some of these results were further improved in~\cite{Censor-HillelLP21,Censor-HillelLP21b} where it is shown how to solve multiple SSSP problems  exactly in $\widetilde{O}(n^{\nicefrac13})$ rounds, and how to approximate SSSP in $\widetilde{O}(n^{\nicefrac{5}{17}})$ rounds. Other graph problems, such as spanning tree, maximal independent set (MIS) construction, and routing were also  considered in the  \local+\ncc\/ model (see \cite{CoyC0HKSSS21,GotteHSW21}). In fact, it was very recently shown~\cite{AnagnostidesG21} that any problem on sparse graphs can be solved in $\widetilde{O}(\sqrt{n})$ rounds in the \local+\ncc\/ model. Efficient distributed algorithms for general graphs in this model can then be obtained using sparsification techniques. Finally, it is worth pointing out that the weaker hybrid model \congest+\ncc\/ was considered in~\cite{FHS20} for restricted families of graphs. 


As a final remark, it is interesting to notice that the {\tt XOR-Index} problem is related to the EPR paradox~\cite{EPR35}, and especially the so-called CHSH game~\cite{CHSH69} whose objective is to demonstrate the existence of quantum (non-classical) correlations in physics (see~\cite{ArfaouiF14}).

\section{Hybrid Models Based on \local\/ and \bcc}

In this section, we consider the combination of \local\/ and \bcc, and, in particular, we compare the two classed $\L\B$ and $\B\L$. The section can be considered as a warmup section before stating more complex separation results further in the text. 

First, we establish a general result concerning the hybridation of  \local\/ and \bcc. Recall that $\prod_{i=1}^k \L^{\alpha_i}\B^{\beta_i}$ is the class of distributed decision problems that can be decided by an algorithm performing $\alpha_1$ rounds of \local, then $\beta_1$ rounds of \bcc, then $\alpha_2$ rounds of \local, etc., ending with $\beta_k$ rounds of \bcc. We show that every language in this class can be computed in the same number of rounds by performing first all \local\/ rounds, and then all \bcc\/ rounds. 

\begin{theorem}\label{theo:LalphaBbeta}
Let $k\geq 1$ be an integer, and let  $\alpha_1,\dots,\alpha_k$ and $\beta_1,\dots,\beta_k$ be non-negative integers. We have 
$\prod_{i=1}^k \L^{\alpha_i}\B^{\beta_i} \subseteq \L^{\sum_{i=1}^k\alpha_i}\B^{\sum_{i=1}^k\beta_i}$.
\end{theorem} 

\begin{proof}
Let $\DL\in \prod_{i=1}^k \L^{\alpha_i}\B^{\beta_i}$, and let $A$ be a distributed algorithm deciding $\DL$ in the corresponding hybrid model combining \local\/ and \bcc. Let us consider the maximum integer ${t< \sum_{i=1}^k(\alpha_i+\beta_i)}$ such that $A$ performs \bcc\/ at round~$t$, and \local\/ at round $t+1$. (If no such $t$ exist, then $A$ is already in the desired form.) We transform $A$ into $A'$ performing the same as~$A$, excepted that rounds $t$ and $t+1$ are switched. Specifically, let us consider a run of $A$ for an instance $(G,\ell)$. Let $B_u$ be the message broadcasted by~$u$ at round~$t$ of~$A$, and, for every neighbor $v$ of~$u$, let $L_{u,v}$ be the message sent by~$u$ to~$v$ at round~$t+1$ of~$A$. To define $A'$, let $S_u$ be the state of every node~$u$ at the beginning of round~$t$ of~$A$, and let $N_G(u)$ be the set of neighbors of~$u$ in~$G$. In $A'$, every node $u$ sends its state~$S_u$ to all its neighbors at round~$t$, using \local. At round $t+1$ of~$A'$, every node~$u$ broadcasts $B_u$ to all nodes, using \bcc\/ (this is doable, as $u$ was able to produce~$B_u$ based on~$S_u$ at round~$t$). Finally, before completing round $t+1$, every node $u$ uses the collection $\{S_v:v\in N_G(u)\}$ and the collection $\{B_w : w\in V(G)\}$ to compute the messages $L_{v,u}$ for all $v\in N_G(u)$, by simulating what every such neighbor would have done~$v$ at round~$t$ of~$A$. Indeed, $L_{v,u}$ depend solely on $S_v$ and $\{B_w : w\in V(G)\}$. (We make the standard assumption that all nodes are running the same algorithm, but even if that was not the case, every node could also send the code of its algorithm to all its neighbors together with its state at round~$t$.) It follows that, at the end of round $t+1$ of~$A'$, every node~$u$ can compute its state after $t+1$ rounds of~$A$. By repeating the same switch operation until no \local\/ rounds occur after a \bcc\/ round, we eventually obtain an algorithm deciding~$\DL$ and establishing that $\DL\in\L^{\sum_{i=1}^k\alpha_i}\B^{\sum_{i=1}^k\beta_i}$. 
\end{proof}

\begin{corollary}
$\B\L\subsetneq \L\B$.
\end{corollary} 

\begin{proof}
The fact that $\B\L\subseteq \L\B$ is a direct consequence of Theorem~\ref{theo:LalphaBbeta}. On the other hand, there is a distributed language in $\L\B\smallsetminus\B\L$ since, as shown by Theorem~\ref{theo:BCvsCB}, $\C\B\smallsetminus\B\L^*\neq \varnothing$.
\end{proof}

We now show a separation between the class $\B\L$ and the class $\B^*\cup\L^*$ of languages that can be decided in a constant number of rounds either in \bcc\/ or \local. The proof does not use communication complexity reduction, but a mere reduction to {\tt triangle-freeness}. 

\begin{theorem}\label{theo:BLvsByL}
$\B\L \smallsetminus (\B^* \cup \L^*) \neq \varnothing$
\end{theorem}

\begin{proof}
Let us consider the language \texttt{triangle-on-max-degree-freeness (TOMDF)} defined by the set of graphs $G$ such that, for every triangle $T$ in~$G$, all nodes in $T$ have a degree smaller than the maximum degree of $G$. 
Note that $\mbox{\tt TOMDF}\in \B\L$. Indeed, during the \bcc\/ round, every node can broadcast its degree. Thus, during the \local\/ round, each node can learn all triangles it belongs to. Every node rejects if it is of maximum degree, and it is contained in a triangle. Otherwise, it accepts. Moreover, $\mbox{\tt TOMDF}\notin\L^*$ because, for every $k\geq 0$,  in $k$ \local\/ rounds a node cannot distinguish an instance~$G$ in which it has maximum degree from an instance~$G'$ in which there is a node with a larger degree. It remains to show that $\mbox{\tt TOMDF}\notin\B^*$. 

Let us assume, for the purpose of contradiction, that there exists $k\geq 0$, such that {\tt TOMDF} can be decided by an algorithm $\mathcal{A}$ performing $k$ \bcc\/ rounds, i.e., $\mbox{\tt TOMDF} \in\B^k$. We can use $\mathcal{A}$ to decide {\tt triangle-freeness} in $k+1$ \bcc\/ rounds. Let $G$ be a graph. In the first \bcc\ round, every node $v$ broadcasts its identifier $\id(v)$ and its degree~$d(v)$, and hence learns the maximum degree $\Delta$ of~$G$. Then every node simulates $\mathcal{A}$ on the virtual graph $G'$ on $\frac{n \Delta}2$ nodes obtained from $G$ by adding a set $S_v$ of $\Delta - d(v)$ pending vertices to each vertex~$v$ of~$G$. Every node $v$ simulates $\mathcal{A}$ in $G'$ by simulating its execution on~$v$ and on all the nodes in~$S_v$. Specifically, after the first \bcc\/ round, $v$~knows the set of IDs used in~$G$, and thus the rank of its ID in this set. Therefore, it can compute the set $I$ composed of the smallest  $\frac{n \Delta}2-n$ positive integers that are not used as IDs in~$G$. Furthermore, it can assign IDs to its $\Delta - d(v)$ pending virtual neighbors in~$G'$, using its rank and the degrees of all the nodes with lower rank in~$G$, so that (1)~the ID of each virtual node is unique in~$G'$, and (2)~every node of $G$ knows the IDs assigned to the pending virtual neighbors of every other node in~$G$.  It follows that each node $v$ does not need to simulate the messages broadcasted in $\mathcal{A}$  by the nodes in~$S_v$. In fact, every node~$v$ can simulate the behavior of all the virtual nodes in $S=\cup_{u\in V(G)}S_u$ at each round of~$\mathcal{A}$. As a consequence, the simulation of $\mathcal{A}$ in $G'$ does not yield any overhead on the number of bits  to be broadcasted by each (real) node~$v$ running~$\mathcal{A}$. After the $k$ \bcc\/ rounds of $\mathcal{A}$ in $G'$ have been simulated, every node~$v$ accepts (on~$G$) if itself and all the nodes in $S_v$ accept in $\mathcal{A}$ on~$G'$. Now, by construction, $G'\in \mbox{\tt TOMDF}$ if and only if $G$ is triangle-free. Since $\mathcal{A}$ decides {\tt TOMDF}, we get that $\mbox{\tt triangle-freeness}\in \B^{k+1}$, a contradiction. 
  \end{proof}

\section{Hybrid Models Based on {\bcc} and {\congest}}

In this section, we consider the combination of \congest\/  and \bcc, and, in particular, we compare the two classes $\C\B$ and $\B\C$. The separation of these two classes uses the communication complexity problem {\tt XOR-Index}. In the next section, we'll establish that the 2-ways 1-round communication complexity of  {\tt XOR-Index} is $\Omega(n)$ bits. We use this lower bounds in the proofs of this section. 


We first show that not only $\C\B\smallsetminus\B\C\neq\varnothing$ but also $\C\B\smallsetminus\B\L^*\neq\varnothing$.

\begin{theorem}\label{theo:CBvsBL}
$\C\B\smallsetminus \B\L^* \neq \varnothing$. This result holds even for randomized algorithms performing one \bcc\/ round followed by a constant number of \local\/ rounds, which may err with probability $\epsilon$, for every $\epsilon<\nicefrac15$.
\end{theorem}

\begin{proof}
Let us consider the distributed language denoted {\tt one-marked-edge} defined as 
\begin{align*}
\mbox{\tt one-marked-edge} = \Big \{(G, \ell) :&\;  \big (\ell:V(G)  \rightarrow \{0,1\} \big) \\
& \land \big(\big|\big\{\{u,v\}\in E(G):\ell(u)=\ell(v)=1\big\}\big| = 1\big)\Big \}.
\end{align*}
In words, the language corresponds to the graphs $G$ with a potential mark on each node, satisfying that exactly one edge of $G$ has its two endpoints marked. We have ${\mbox{\tt one-marked-edge}\in\C\B}$. Indeed, a simple algorithm consists, for each node, to learn which of its neighbors are marked, in one {\congest} round, and to broadcast its number of marked incident edges, in one {\bcc} round. The nodes reject if the total sum of marked edges is different from~2 (i.e., exactly two nodes are incident to a unique marked edge). They accept otherwise. 

We now prove that $\mbox{\tt one-marked-edge}\notin \B\L^*$. We show that this result holds even for a randomized algorithm which may err with probability $\epsilon<\nicefrac15$.  For the purpose of contradiction, let us assume that, for some $k \geq 0$, there exists an $\epsilon$-error algorithm $\mathcal{A}$ solving {\tt one-marked-edge} using one {\bcc} round followed by $k$ consecutive {\local} rounds. We show how to use $\mathcal{A}$ for designing  an $\epsilon$-error 1-round protocol $\Pi$ solving $\xorindex$ by communicating only $\cO(\sqrt{m})$ bits on $m$-bit instances, contradicting the fact that $\xorindex$ has communication complexity~$\Omega(m)$. 


Let $(x,i)\in \{0,1\}^m \times [m]$ and $(y,j) \in \{0,1\}^m \times [m]$ be an instance of $\xorindex$. Without loss of generality, we assume that $ m = {n \choose 2}$ for some $n\in \mathbb{N}$. Let us consider a graph $G$ on $2n + 4k$ nodes, composed of two disjoint copies of a clique of size~$n$, plus a path $P$ of $4k$ nodes. Let us denote by $G^A$ and $G^B$ the two cliques. The IDs assigned to the nodes of $G^A$ are picked in $[1,n]$, while the IDs assigned to the nodes of $G^B$ are picked in $[n+1,2n]$. One extremity of $P$ is connected to all nodes in~$G^A$, and  the other extremity of $P$ is connected to all nodes in~$G^B$. Let us denote by $P^A$ the $2k$ nodes of $P$ closest to $G^A$, and by $P^B$ the $2k$ nodes of $P$ closest to $G^B$. These nodes are assigned IDs $2n+1,\dots,2n+4k$, consecutively, starting from the extremity of~$P$ connected to~$G^A$. 

We enumerate the $ m = {n \choose 2}$ edges in $G^A$ and $G^B$ from $1$ to $m$. Then, in $\Pi$, the players interpret their input vectors $x$ and $y$ as indicators of the edges of $G^A$ and $G^B$ respectively. We denote by $G_{xy}$ the  subgraph of $G$ such that, for every $r\in [m]$, the $r$-th edge $e$ of $G^A$ (resp.,~$G^B$) is in $G_{xy}$ if and only if $x_r = 1$ (resp., $y_r = 1$). Also, all edges incident to nodes of $P$ are in~$G_{xy}$.  Let $\{u^i_A,v^i_A\}$ be the endpoints of the $i$-th edge of $G_A$, and let $\{u^j_B, v^j_B\}$ represent the endpoints of the $j$-th edge of $G_B$. (These edges may or may not be in~$G_{xy}$ depending on the values of $x_j$ and $y_i$.) We define $\ell_{ij} : V(G) \rightarrow \{0,1\}$ as the marking function such that $\ell_{ij}(w) = 1$ if and only if $w \in \{u^i_A, v^i_A, u^j_B , v^j_B\}$. By construction, we have that $(G_{xy}, \ell_{ij})\in \mbox{\tt one-marked-edge}$ if and only if $((x,i),(y,j))$ is a yes-instance of \xorindex, i.e., $x_j\neq y_i$. We say that Alice owns all nodes in $V(G^A)\cup V(P^A)$, and Bob owns all nodes in $V(G^B)\cup V(P^B)$. Observe that the edges of $G_{xy}$ incident to nodes owned by Alice depend only on~$x$, while the edges of $G_{xy}$  incident to nodes owned by Bob only depend on~$y$. 

We are now ready to describe $\Pi$. First, Alice and Bob simulate the {\bcc} round of algorithm $\mathcal{A}$ on all the nodes  of $G_{xy}$ owned by them, respectively, considering that \emph{no vertices are marked}. This simulation results in each player constructing a set of $n+2k$ messages, one for each node of the  clique owned by the player, plus one message for each of the $2k$ nodes  in the sub-path owned by the player. We denote by $M^A_0$ and $M^B_0$ the set of messages produced by Alice and Bob, respectively. Next, the players repeat the same procedure, but considering now that \emph{all vertices are marked}, from which it results sets of messages denoted by $M^A_1$ and $M^B_1$, respectively. Finally, Alice sends the pair $(M^A_0,M^A_1)$ to Bob, as well as her input index~$i$. Similarly, Bob sends the pair $(M^B_0,M^B_1)$ to Alice, as well as  his input index~$j$. Observe that the size of these messages is $\cO((n+k) \log n)$ bits. 

After the communication, Alice and Bob decide their outputs as follows. First, each player extracts from $M_1^A$ the messages produced by $u^j_A$ and~$v^j_A$, and extract from $M_1^B$ the messages produced by $u^i_B$ and~$v^i_B$. Then, they extract from $M_0^A$ and $M_0^B$ the messages of every other node. Let us call $M$ the resulting set of messages. Observe that $M$ corresponds exactly to the set of messages communicated during the {\bcc} round of $\mathcal{A}$ on input $(G_{xy}, \ell_{ij})$. Then, Alice and Bob simulate the $k$ {\local} rounds of $\mathcal{A}$ on all the vertices they own. This is possible as the nodes of $P$ are not marked, for every instance of \xorindex. Each player accepts if all the nodes owned by this player accept. 
Since $(G_{xy}, \ell_{ij})\in \mbox{\tt one-marked-edge}$ if and only if $((x,i),(y,j))$ is a yes-instance of {\xorindex}, we get that $\Pi$ is an $\epsilon$-error protocol solving $\xorindex$ on inputs of size~$m$ by  communicating only $\cO((n+k) \log n) = \cO(\sqrt{m})$ bits, which is a contradiction with Theorem~\ref{theo:xorindex}.
\end{proof}


We now show that $\B\C\smallsetminus\C\B\neq\varnothing$.

\begin{theorem}\label{theo:BCvsCB}
 $\B\C \smallsetminus \C\B \not = \varnothing$. This result holds even for randomized algorithms performing one \congest\/ round followed by one \bcc\/ round, which may err with probability $\epsilon$, for every~$\epsilon<\nicefrac15$.
\end{theorem}

\begin{proof}
 For every $n\geq 2$, let us consider the path $P_{2n+1}$, i.e., the path with $2n+1$ nodes, denoted consecutively $a_1,\dots,a_n,c,b_n,\dots,b_1$. Let $x\in \{0,1\}^n$,  $y\in \{0,1\}^n$, $i\in[n]$, and $j\in[n]$. We define the labeling $\ell_{x,y,i,j}$ of the nodes of~$P_n$ as follows: 
\[
\ell_{x,y,i,j}(a_1) = i, \; \;
\ell_{x,y,i,j}(a_n) = x, \; \; 
\ell_{x,y,i,j}(b_n) = y, \; \;
\ell_{x,y,i,j}(b_1) = j,
\]
and, for every $v\notin \{a_1,a_n,b_1,b_n\}$, $\ell_{x,y,i,j}(v) = \bot$. We define the distributed language 
\[
\mbox{\tt XOR-index-path} = \{(P_{2n+1},\ell_{x,y,i,j}): (n\geq 2) \land (x, y \in \{0,1\}^n)\land (i,j \in [n])\land  (x_j\neq y_i)\}. 
\]
First, we show that $\mbox{\tt XOR-index-path}\in\B\C$.  During the \bcc\/ round, every node broadcasts its ID, and the IDs of its neighbors (a node with more than two neighbors simply rejects). Also, degree-1 nodes broadcasts their labels. Note that the $2n+1$ nodes can then check whether they are vertices of the path~$P_{2n+1}$, and, if this is not the case, they reject. Let $i$ and $j$ be the labels broadcasted by the two extremities of the path. Based on the information broadcasted by all the nodes, each of the two nodes $a_n$ and $b_n$ adjacent to the middle node~$c$ of the path knows which of the two labels $i$ or $j$ correspond to the index broadcasted by its farthest extremity in the path, $b_1$ and $a_1$, respectively. Thus, during the \congest\/ round, $a_n$ and $b_n$ can send the bits $x_j$ and $y_i$ to the center~$c$ of the path, which checks whether $x_j\neq y_i$, and accepts or rejects accordingly. 

Now, we show that $\mbox{\tt XOR-index-path} \notin \C\B$.  Let us assume for the purpose of contradiction that there exists a $2$-round algorithm $\mathcal{A}$ deciding $\mbox{\tt XOR-index-path}$ by performing one \congest\/ round followed by one \bcc\/ round. To solve an instance $((x,i),(y,j))$ of {\tt XOR-Index}, Alice and Bob simulate $\mathcal{A}$ on the path $P_{2n+1}$ with consecutive IDs $1,\dots,2n+1$. Specifically, Alice simulates the $n+1$ nodes $a_1,\dots,a_n,c$, while Bob simulates the $n+1$ nodes $b_1,\dots,b_n,c$, with the nodes labeled with~$\ell_{x,y,i,j}$. For simulating the \congest\/ round, Alice sends to Bob the message $m_{a_n}$ sent from $a_n$ to~$c$ during that round, and Bob sends to Alice the message $m_{b_n}$ sent from~$b_n$ to~$c$ during that round. The \bcc\/ round is actually simulated simultaneously. More precisely, Alice and Bob can both construct the messages broadcasted by all nodes $a_3,\dots,a_{n-2}$ and $b_3,\dots,b_{n-2}$, merely because they know their IDs and their labels (equal to~$\bot$), and they can therefore infer the messages these nodes receive during the \congest\/ round. So, these messages do not need to be communicated between the players. Moreover, Alice knows a priori what messages $m'_{a_1},m'_{a_2}$, and $m'_{a_n}$ are to be broadcasted by $a_1,a_2$ and $a_n$ during the \bcc\/ round, and can send them to Bob. Symmetrically,  Bob knows a priori what messages $m'_{b_1},m'_{b_2}$, and $m'_{b_n}$ are to be broadcasted by $b_1,b_2$ and $b_n$ during the \bcc\/ round, and can send them to Alice. As for node~$c$, thanks to the messages $m_{a_n}$ and $m_{b_n}$ sent by Alice to Bob, and by Bob to Alice, respectively, both players can construct the message to be sent by~$c$ during the \bcc\/ round. So, in total, for simulating~$\mathcal{A}$, Alice (resp., Bob) just needs to send the messages $m_{a_n}, m'_{a_1},m'_{a_2},m'_{a_n}$  to Bob (resp., the messages $m_{b_n}, m'_{b_1},m'_{b_2},m'_{b_n}$  to Alice), which consumes $O(\log n)$ bits of communication in total. Each player accepts if all the nodes he or she simulates accept, and rejects otherwise. Alice and Bob are thus able to solve {\tt XOR-index} by exchanging $\mathcal{O}(\log n)$ bits only, which contradicts Theorem \ref{theo:xorindex}. 
\end{proof}

As a direct consequence of the previous two theorems, we get: 

\begin{corollary}
 The sets $\C\B$ and $\B\C$ are incomparable.
\end{corollary}



\section{The Communication Complexity of \xorindex}

This section is dedicated to the analysis of the following communication problem. 

\medbreak 
\fbox{
\begin{minipage}{13cm}
\noindent \xorindex: 
\begin{description}
\item[Input:] Alice receives $x\in \{0,1\}^n$ and $i \in [n]$; Bob receives $y \in \{0,1\}^n$ and $j \in [n]$. 
\item [Task:] Alice outputs a boolean $\textrm{out}_A$ and Bob outputs a boolean $\textrm{out}_B$ such that $\textrm{out}_A \wedge \textrm{out}_B = x_i  \oplus  y_i.$
\end{description}
\end{minipage}
}
\medbreak 

We focus on 2-way 1-round protocols, that is, each player sends only one message to the other player, both players send their messages simultaneously, and each player must decide his or her output upon reception of the message sent by the other player.
For every 2-player communication problem~$P$, and for every $\epsilon>0$, let us denote by $CC^1(P,\epsilon)$ the communication complexity of the best 2-way 1-round randomized protocol solving $P$ with error probability at most~$\epsilon$.  

\begin{theorem} \label{theo:xorindex}
For every non-negative $\epsilon<\nicefrac15$, $CC^1(\mbox{\rm \xorindex},\epsilon) = \Omega(n)$ bits.
\end{theorem}

The rest of the section is entirely dedicated to the proof of Theorem~\ref{theo:xorindex}. Let $0\leq \epsilon<\nicefrac15$, and let $\Pi$ randomized protocol  solving \xorindex\/ with error probability at most~$\epsilon$, where Alice communicates $k_A$ bits to Bob, and Bob communicates $k_B$ bits to Alice. Without loss of generality, we can assume that, in $\Pi$, Alice (resp., Bob) sends explicitly the value of~$i$ (resp.,~$j$) to Bob (resp., Alice). Indeed, this merely increases the communication complexity of~$\Pi$ by an additive factor $O(\log n)$, which has no consequence, as we shall show that $k_A+k_B=\Omega(n)$. 

Let us consider the probabilistic distribution over the inputs of Alice and Bob, where $x$ and $y$ are drawn uniformly at random from $\{0,1\}^n$, and $i$ and $j$ are drawn uniformly at random from~$[n]$. Let us denote $X$ and $I$ the random variables equal to the inputs of Alice, and $Y$ and $J$ the random variables equal to the inputs of Bob. Let $M_A$ (resp.,~$M_B$) be the random variable equal to the message sent by Alice (resp., Bob) in $\Pi$ on input $(X,I)$ (resp., $(Y,J)$). Note that $M_A$ and $M_B$ have values in $\Omega_A=\{0,1\}^{k_A}$ and $\Omega_B=\{0,1\}^{k_B}$, respectively, of respective size $2^{k_A}$ and  $2^{k_B}$.

Let us fix $i, j \in [n]$, $m_A \in \Omega_A$, and $m_B \in \Omega_B$.  Let  $\mathcal{E}^A_{m_A, j}$ be the event corresponding to Bob receiving $J=j$ as input, and Alice sending $M_A = m_A$ to Bob in the communication round. Similarly, let $\mathcal{E}^B_{m_B, i}$ be the event corresponding to Alice receiving  $I=i$ as input, and Bob sending $M_B = m_B$ to Alice in the communication round. For $a,b\in\{0,1\}$, we set: 
\[
p(a,m_A,j) = \Pr[X_J = a  \mid \mathcal{E}^A_{m_A,j} ], \;\mbox{and}\;
q(b,m_B,i) = \Pr[ Y_I = b  \mid \mathcal{E}^B_{m_B,i}], 
\]
and 
\[
p(a,j) = \Pr[ X_J = a  \mid  J = j ], \;\mbox{and}\;
q(b,i) = \Pr[ Y_I = b  \mid I = i ]. 
\]
Observe that $p(a,j) = q(b,j) = 1/2$. Let $a^*$ and $b^*$ be the most probable values of $X_j$ given $(m_A, j)$, and of  $Y_i$ given $(m_B, i)$, respectively. Formally,
\[
a^* = \textrm{argmax}_{a\in \{0,1\}} p(a,m_A,j), \; \;\mbox{and}\; \;
b^* = \textrm{argmax}_{b\in \{0,1\}} q(b,m_B,i).
\]
Observe that $p(a^*,m_A, j) \geq 1/2$ and $q(b^*,m_A, j) \geq 1/2$. We first establish the following technical lemma. 

\begin{lemma}\label{lem:xorindextechnical}
Let $\mathcal{F}$ the the event that $\Pi$ fails. We have 
\[
\Pr[ \mathcal{F} \mid \mathcal{E}_{m_A,j}^A, \mathcal{E}_{m_B,i}^B]  \geq 
 \Pr[ a^* \neq  X_J \mid \mathcal{E}_{m_A,j}^A, \mathcal{E}_{m_B,i}^B] \cdot  
 \Pr[b^* \neq Y_I  \mid \mathcal{E}_{m_A,j}^A, \mathcal{E}_{m_B,i}^B].
 \] 
\end{lemma}

\begin{proof}
Without loss of generality, we assume that, in $\Pi$, after having communicated the pair $(m_A,i)$, Alice computes $b^*$, and decides her output as follows. If $b^* \neq x_j$, then Alice accepts with some fixed probability~$p_A$, and if $b^* = x_j$ then Alice accepts with some fixed probability~$q_A$. The probabilities $p_A$ and $q_A$ determines the actions of Alice.  Similarly, we can assume that, after having communicated $(m_B,j)$, Bob computes $a^*$, and decides as follows. If $a^* \neq y_i$ then he accepts with some fixed probability $p_B$, and if $a^* = y_i$ then he accepts with some fixed probability~$q_B$. Note that, in the case where the players do not take in account the value of $a^*$ and $b^*$, then one can simply choose $p_A = q_A$ and $p_B = q_B$. 
Let us denote 
\[
R_A = \Pr[a^* = X_J \mid \mathcal{E}_{m_A,j}^A, \mathcal{E}_{m_B,i}^B], \;\; \mbox{and} \;\;
R_B = \Pr[b^* = Y_I \mid \mathcal{E}_{m_A,j}^A, \mathcal{E}_{m_B,i}^B].
\]
 Observe that 
\[
\Pr[\overline{\mathcal{F}} \mid \mathcal{E}_{m_A,j}^A, \mathcal{E}_{m_B,i}^B]
= \nicefrac{1}{2} \Pr[\overline{\mathcal{F}} \mid \mathcal{E}_{m_A,j}^A, \mathcal{E}_{m_B,i}^B, X_J \neq Y_I]
+ \nicefrac{1}{2}  \Pr[\overline{\mathcal{F}} \mid \mathcal{E}_{m_A,j}^A, \mathcal{E}_{m_B,i}^B, X_J = Y_I].
\]
Now, conditioned on $X_J \neq Y_I$, the event $\overline{\mathcal{F}}$ corresponds to the event when Alice accepts and Bob accepts. Observe that, conditioned on $\mathcal{E}_{m_A,j}^A, \mathcal{E}_{m_B,i}^B$,  these two latter events are independent.  
Moreover, conditioned on  $X_J \neq Y_I$, the event $a^* \neq Y_I$ is equal to the event $a^* = X_J$. Similarly, conditioned on $X_J \neq Y_I$, the event $b^* \neq X_J$ is equal to the event $b^* = X_I$. It follows that 
\[
\left\{\begin{array}{ll}
\Pr[ \textrm{Alice accepts} \mid \mathcal{E}_{m_A,j}^A, \mathcal{E}_{m_B,i}^B , X_J \neq Y_I]= R_B\, p_A + (1-R_B)\, q_A;\\
\Pr[\textrm{Bob accepts}\mid \mathcal{E}_{m_A,j}^A, \mathcal{E}_{m_B,i}^B , X_J \neq Y_I] = R_A\,p_B + (1-R_A)\, q_B.
\end{array}\right.
\]
This implies that 
\begin{align*} 
\Pr[&\overline{\mathcal{F}} \mid \mathcal{E}_{m_A,j}^A, \mathcal{E}_{m_B,i}^B, X_J \neq Y_I]\\ 
& = \Pr[\textrm{Alice accepts and Bob accepts } \mid \mathcal{E}_{m_A,j}^A, \mathcal{E}_{m_B,i}^B, X_J \neq Y_I] \\
& = \Pr[ \textrm{Alice accepts} \mid \mathcal{E}_{m_A,j}^A, \mathcal{E}_{m_B,i}^B , X_J \neq Y_I] 
      \cdot \Pr[\textrm{Bob accepts} \mid \mathcal{E}_{m_A,j}^A, \mathcal{E}_{m_B,i}^B , X_J \neq Y_I]\\
& = \big (R_B\, p_A + (1-R_B)\, q_A\big ) \cdot \big (R_A\, p_B + (1-R_A)\, q_B\big ).
\end{align*}
let us now consider the case when conditioning on $X_J = Y_I$.  In this case, the event $\overline{\mathcal{F}}$ corresponds to the complement of the event when Alice accepts and Bob accepts. Observe that, conditioned on $X_J = Y_I$, the event $a^* \neq Y_I$ is equal to the event $a^*\neq X_J$, and the event $b^* \neq X_J$ is equal to the event $b^* \neq Y_I$. It follows that 
\[
\left\{\begin{array}{ll}
\Pr[\textrm{Alice accepts} \mid \mathcal{E}_{m_A,j}^A, \mathcal{E}_{m_B,i}^B , X_J = Y_I] = (1-R_B)\, p_A + R_B\, q_A;\\
\Pr[\textrm{Bob accepts} \mid \mathcal{E}_{m_A,j}^A, \mathcal{E}_{m_B,i}^B , X_J = Y_I] = (1-R_A)\, p_B + R_A\, q_B/
\end{array}\right.
\]
This implies that 
\begin{align*} 
\Pr[&\overline{\mathcal{F}} \mid \mathcal{E}_{m_A,j}^A, \mathcal{E}_{m_B,i}^B, X_J = Y_I)]\\ 
& = 1 - \Pr[\textrm{Alice accepts and Bob accepts } \mid \mathcal{E}_{m_A,j}^A, \mathcal{E}_{m_B,i}^B, X_J = Y_I]\\
& = 1 - \Pr[\textrm{Alice accepts} \mid \mathcal{E}_{m_A,j}^A, \mathcal{E}_{m_B,i}^B , X_J = Y_I] 
      \cdot \Pr[\textrm{Bob accepts} \mid \mathcal{E}_{m_A,j}^A, \mathcal{E}_{m_B,i}^B , X_J = Y_I]\\
& = 1-  \big ((1-R_B)\, p_A + R_B\, q_A\big ) \cdot \big((1-R_A)\, p_B + R_A\, q_B\big ). 
\end{align*}
Therefore, by combining the two cases, we get that
\begin{align*}
\Pr[&\overline{\mathcal{F}} \mid \mathcal{E}_{m_A,j}^A, \mathcal{E}_{m_B,i}^B]\\ 
&= \frac{1}{2} \big( R_A\, (p_A+q_A)\, (p_B-q_B) + R_B\, (p_A-q_A)\, (p_B+q_B) + 1 - p_A\, p_B + q_A\, q_B  \big).
\end{align*}

Conditioned to the events $\mathcal{E}_{m_A,j}^A, \mathcal{E}_{m_B,i}^B$, the best protocol $\Pi$ corresponds to the one that picks the values of $p_A, q_A, p_B, q_B$ that maximize the previous quantity, restricted to the fact that  $p_A, q_A, p_B, q_B, R_A$ and $R_B$ must be values in $[0,1]$, and that $R_A$ and $R_B$ must be at least $\nicefrac12$. The maximum can be found using the Karush-Kuhn-Tucker (KKT) conditions~\cite{sundaram1996first}. In fact, as the restrictions are affine linear functions, the optimal value is one solution of the following system of equations:

$$\begin{array}{ccc}
{\left(R_{A} + R_{B} - 1\right)} p_{B} - {\left(R_{A} - R_{B}\right)} q_{B} - 2\mu_{1} + 2\mu_{5} &=& 0\\
 {\left(R_{A} + R_{B} - 1\right)} p_{A} +  {\left(R_{A} - R_{B}\right)} q_{A} - 2\mu_{2} + 2\mu_{6} &=& 0\\
 {\left(R_{A} - R_{B}\right)} p_{B} -  {\left(R_{A} + R_{B} - 1\right)} q_{B} - 2\mu_{3} + 2\mu_{7} &=& 0\\
- {\left(R_{A} - R_{B}\right)} p_{A} - {\left(R_{A} + R_{B} - 1\right)} q_{A} - 2\mu_{4} + 2\mu_{8} &=& 0\\
\mu_{1} {\left(p_{A} - 1\right)} &=& 0 \\
\mu_{2} {\left(p_{B} - 1\right)} &=& 0 \\
\mu_{3} {\left(q_{A} - 1\right)} &=& 0 \\
\mu_{4} {\left(q_{B} - 1\right)} &=& 0 \\
-\mu_{5} p_{A} &=& 0 \\
-\mu_{6} p_{B} &=& 0 \\
-\mu_{7} q_{A} &=& 0 \\
-\mu_{8} q_{B} &=& 0 \\
\end{array}$$
The set of all solutions to this system is given in Table \ref{table:KKT}, together with the corresponding evaluation of $\Pr[\overline{\mathcal{F}}\mid \mathcal{E}_{m_A,j}^A, \mathcal{E}_{m_B,i}^B)$. 
It follows from Table \ref{table:KKT} that the value of $\Pr[\overline{\mathcal{F}}\mid \mathcal{E}_{m_A,j}^A, \mathcal{E}_{m_B,i}^B)$ is upper bounded by $R_A + R_B - R_AR_B$. Indeed, assuming, w.l.o.g., that $R_A \geq R_B$, we have:
\begin{align*}
\frac{1}{2}( 1 - (R_A + R_B)) &  \leq 1-R_A \leq 1 - \frac{R_A + R_B}{2} \leq 1 - R_B \leq \frac{1}{2}\\
& \leq \min\{R_B, \frac{1}{2}(1 + R_A - R_B)\} \leq \max\{R_B, \frac{1}{2}(1 + R_A - R_B)\}\\
& \leq \frac{R_A + R_B}{2} \leq R_A \leq 1 - (1-R_A)(1-R_B). 
\end{align*}
Finally, observe that 
\[
 (1-R_A)(1-R_B) 
 =  \Pr[a^* \neq  X_J \mid \mathcal{E}_{m_A,j}^A, \mathcal{E}_{m_B,i}^B] 
\cdot  \Pr[b^* \neq Y_I \mid \mathcal{E}_{m_A,j}^A, \mathcal{E}_{m_B,i}^B], 
\]
from which we get that  
\[
\Pr[\mathcal{F} \mid  \mathcal{E}_{m_A,j}^A, \mathcal{E}_{m_B,i}^B]  
\geq  \Pr[ a^* \neq  X_J |\mathcal{E}_{m_A,j}^A] \cdot  \Pr[b^* \neq Y_I \mathcal{E}_{m_B,i}^B],  
\]
as claimed. 
\end{proof}

We now show that, whenever the messages sent by Alice and Bob are too small, the distributions of $a^*$ and of $b^*$ is not far from the uniform. We make use of some basic definitions and tools on information complexity, and we refer to~\cite{rao2020communication} for more details. Let $(\Omega, \mu)$ be a discrete probability space. Given a random variable $X$ we denote by $p_X: \Omega \mapsto \mathbb{R}$ the discrete density function of $X$, i.e., $p_X(\omega) = \Pr[X =\omega]$. We denote by $\mathbb{H}: \Omega \mapsto \mathbb{R}^{+}$ the entropy function, defined as $\mathbb{H}(X) = \sum_{\omega \in \Omega} p_X(\omega) \frac{1}{\log p_X(\omega)}.$ Recall that, given two random variables $X,Y$ on $\Omega$, the entropy of $X$ conditioned to $Y$ is 
\[
\mathbb{H}(X\mid Y) = \mathbb{E}_{p_Y(y)}(H(X\mid Y=y)).
\] 
Moreover, let $\mu$ and $\nu$ be two probability measures on $\Omega$. The total variation distance between $\mu$ and $\nu$ is defined as $|u-v|_{\text{\tiny TV}} = \sup_{E\subseteq \Omega} |\mu(E) - \nu(E)|$. It is known that  $|u-v|_{\text{TV}} = \frac{1}{2}\sum_{\omega \in \Omega} |\mu(\omega)-\nu(\omega)|$. In addition,  the Kullback-Liebler divergence between $\mu$ and $\nu$ is defined as $\mathbb{D}(\mu||\nu) = \sum_{\omega \in \Omega} \mu(\omega) \log \frac{\mu(\omega)}{\nu(\omega)}.$ 
Given two random variables $X$ and $Y$, their mutual information is defined as $\mathbb{I}(X;Y) = \mathbb{D}(p_{X,Y}||p_{X}p_{Y}).$ It is known that 
\[
\mathbb{I}(X;Y) = \mathbb{H}(X) - \mathbb{H}(X\mid Y) = \mathbb{H}(Y) - \mathbb{H}(Y\mid X) =\mathbb{I}(Y;X).
\]
Finally, the mutual information of $X,Y$ conditioned on a random variable $Z$ is defined as the function $\mathbb{I}(X;Y\mid Z) = \mathbb{E}_{p_{Z}(z)}[\mathbb{I}(X;Y \mid Z=z)].$ Having all these notions at hand, we shall use the following technical lemmas:

\begin{lemma}[Theorem 6.12 in \cite{rao2020communication}]
Let $A_1, \hdots, A_n$ be independent random variables, and let $B$ be
 jointly distributed. We have 
$\sum_{i=1}^{n} \mathbb{I}(A_i;B) \leq \mathbb{I}(A_1, \hdots, A_n; B)$. 
\end{lemma}

\begin{lemma}[Pinsker's Inequality, Lemma 6.13 in \cite{rao2020communication}]
Let $\mu, \nu$ be two probability measures over~$\Omega$. We have 
$|u-v|_{\text{\tiny TV}}^2 \leq \frac{2}{\ln 2} \, \mathbb{D}(\mu||\nu)$. 
\end{lemma}

Back into our problem, we observe that:
\[
\left\{\begin{array}{lclclclcl}
\mathbb{I}(X_J ; M_A \mid J) & = & \frac{1}{n} \sum_{j \in [n]} \mathbb{I}(X_j; M_A )&  \leq & \frac{\mathbb{I}(X; M_A)}n & \leq & \frac{\mathbb{H}(M_A)}n & \leq&  \frac{k_A}{n}\\
\mathbb{I}(Y_I ; M_B \mid  I) & =&  \frac{1}{n} \sum_{i \in [n]} \mathbb{I}(Y_i; M_B ) & \leq & \frac{\mathbb{I}(Y; M_B)}n &\leq &\frac{\mathbb{H}(M_B)}n &\leq &\frac{k_B}{n}.
\end{array}\right.
\]
By Pinsker's inequality, it follows that:
\[
\left\{\begin{array}{lcl}
\mathbb{E}_{(m_A,j)} \left( \Vert p(\cdot ,m_A,j) - p(\cdot ,j) \Vert \right) & \leq & \sqrt{\frac{k_A}{n}} \\
\mathbb{E}_{(m_B,i)} \left( \Vert q(\cdot ,m_B,i) - q(\cdot ,i) \Vert \right) & \leq & \sqrt{\frac{k_B}{n}}
\end{array}\right.
\]
These latter bounds imply  that 
\[
\left\{\begin{array}{lcl}
\mathbb{E}_{(m_A,j)}  \left(p(a^* ,m_A,j) \right) & \leq & \frac{1}{2} + \sqrt{\frac{k_A}{n}} \\
\mathbb{E}_{(m_B,i)}  \left(q(b^* ,m_B,i) \right) & \leq & \frac{1}{2} + \sqrt{\frac{k_B}{n}}
\end{array}\right.
\]
Now, from Lemma~\ref{lem:xorindextechnical}, we have that
\begin{align*}
\Pr[ \mathcal{F} \mid  \mathcal{E}_{m_A,j}^A, \mathcal{E}_{m_B,i}^B] 
& \geq p(1 - a^*, m_A, j) \cdot q(1- b^*, m_B, i)\\
& = \big(1 -  p(a^*, m_A, j) \big) \cdot \big( 1- q(b^* , m_B, i)\big).
\end{align*}
As a consequence, we have 
\begin{align*}
\Pr[\mathcal{F}] &= \mathbb{E}_{m_A, m_B, i, j} (\Pr[\mathcal{F} \mid \mathcal{E}_{m_A,j}^A, \mathcal{E}_{m_B,i}^B])\\
&= \sum_{m_A, m_B, i,j } \Pr[\mathcal{F} \mid \mathcal{E}_{m_A,j}^A, \mathcal{E}_{m_B,i}^B]\cdot 
\Pr[\mathcal{E}_{m_A,j}^A, \mathcal{E}_{m_B,i}^B]\\
&\geq  \sum_{m_A, m_B, i,j } \left(1 -  p(a^*_{(m_A,j)}, m_A, j) \right) \left( 1- q(b^*_{(m_B,i)} , m_B, i)\right)
\Pr[\mathcal{E}_{m_A,j}^A, \mathcal{E}_{m_B,i}^B]\\
&= \sum_{m_A, m_B, i,j } \left(1 -  p(a^*_{(m_A,j)}, m_A, j) \right) \left( 1- q(b^*_{(m_B,i)} , m_B, i)\right)
\Pr[\mathcal{E}_{m_A,j}^A]\cdot \Pr[\mathcal{E}_{m_B,i}^B]\\
&= \sum_{m_A, j } \left(1 -  p(a^*_{(m_A,j)}, m_A, j) \right)\Pr[\mathcal{E}_{m_A,j}^A] \cdot 
\sum_{m_B, i } \left( 1- q(b^*_{(m_B,i)} , m_B, i)\right)\Pr[\mathcal{E}_{m_B,i}^B]\\
&= \left (1 - \mathbb{E}_{(m_A,j)}\big(p(a_{(m_a,j)}^* ,m_A,j)\big )\right)  \cdot 
\left (1- \mathbb{E}_{(m_B,i)}\big (q(b_{(m_B,i)}^* ,m_B,i) \big )\right)\\
&\geq \left( \frac{1}{2} - \sqrt{\frac{k_A}{n}} \right) \cdot \left(\frac{1}{2} - \sqrt{\frac{k_B}{n}}\right). 
\end{align*}
Since $\Pr[\mathcal{F}] \leq \epsilon$, we must have $( \frac{1}{2} - \sqrt{\frac{k_A}{n}}) \cdot (\frac{1}{2} - \sqrt{\frac{k_B}{n}}) \leq \varepsilon \leq \nicefrac{1}{5}$, implying that $k_A = \Omega(n)$ or $k_B = \Omega(n)$.

\section{Conclusion}

In this paper, we have performed an extensive study of 2-round hybrid models resulting from mixing \local, \congest, and \bcc, and we obtained a complete picture of the relative power of these models (see Figure~\ref{fig:lattice}). This is a first step toward approaching the minimization problem expressed in Eq.~\eqref{eq:min}, which asks for identifying the best combination of these three models for which there is an algorithm that solves a given distributed decision problem $\DL\in\L^*\B^*$ with a minimum number of rounds, or at minimum cost. Solving this minimization problem appears to be currently out of reach, but this paper provides some knowledge about the computational power of hybrid models. Concretely, a step forward in the direction of solving the problem of Eq.~\eqref{eq:min} would be to determine whether most hybrid models remain incomparable when allowing $t$~rounds for $t>2$. In particular, in the case of hybrid models mixing \local\/ and \bcc, we have shown that one can systematically assume that all \local\/ rounds are performed before all the \bcc\/ rounds. This does not holds for \congest\/ and \bcc, for 2-round algorithms. However, we do not know whether the classes $\prod_{i=1}^k\B^{\beta_i}\C^{\gamma_i}$ and $\prod_{i=1}^k\B^{\beta'_i}\C^{\gamma'_i}$ are systematically incomparable for all distinct sequences $((\beta_i,\gamma_i):i=1,\dots,k)$ and $((\beta'_i,\gamma'_i):i=1,\dots,k)$  such that $\sum_{i=1}^k(\beta_i+\gamma_i)=\sum_{i=1}^k(\beta'_i+\gamma'_i)$. 

The line of research investigated in this paper could obviously be carried out by considering other models as well, in particular other congested clique models like \ucc\/ and \ncc. It is easy to see that $\mathbf{U}$, the class of distributed languages that can be decided in one round in the unicast congested clique, is incomparable with the largest class of models considered in this paper. Namely, $\mathbf{U}\smallsetminus \L^*\B^*\neq \varnothing$ and $ \L^*\B^* \smallsetminus \mathbf{U} \neq \varnothing$. Also, previous work on the hybrid model combining \local\/ and \ncc\/ reveals that computing the diameter of the network cannot be done in a constant number of rounds in this model. Taking this under consideration, it could be interesting to study the class $(\L\mathbf{N})^*$ of distributed languages that can be decided in a constant number of rounds in the hybrid model combining \local\/ and \ncc, where $(\L\mathbf{N})^*=\cup_{t\geq 0}(\L\mathbf{N})^t$, and $\mathbf{N}$ denotes the class of languages decidable in one round in the node-capacitated clique.

\bibliographystyle{plain} 
\bibliography{biblio-BCL}
\newpage 
\appendix

\section{Bounding the probability of failure}

\renewcommand{\arraystretch}{1.2}
\begin{table}[h]
\begin{center}
\begin{tabular}{|c|c|c|c|c|}
\hline
$p_A$ & $p_B$ & $q_A$ & $q_B$ & 	$ \Pr[\mathcal{F}\mid  \mathcal{E}_{m_A,j}^A, \mathcal{E}_{m_B,i}^B]$ \\
\hline\hline
$1$ & $1$ & $1$ & $1$ & $\nicefrac{1}{2}$ \\ \hline
$0$ & $0$ & $0$ & $0$ & $\nicefrac{1}{2}$ \\ \hline
$0$ & $1$ & $0$ & $\frac{R_{A} - R_{B}}{R_{A} + R_{B} - 1}$ & $\nicefrac{1}{2}$ \\ \hline
$0$ & $\frac{R_{A} + R_{B} - 1}{R_{A} - R_{B}}$ & $0$ & $1$ & $\nicefrac{1}{2}$ \\ \hline
$0$ & $1$ & $1$ & $1$ & $1-R_{B}$ \\ \hline
$1$ & $0$ & $\frac{R_{B} - R_{A}}{R_{A} + R_{B} - 1}$ & $0$ & $\nicefrac{1}{2}$ \\ \hline
$\frac{R_{A} + R_{B} - 1}{R_{B} - R_{A}}$ & $0$ & $1$ & $0$ & $\nicefrac{1}{2}$ \\ \hline
$1$ & $0$ & $1$ & $1$ & $1 -R_{A}$ \\ \hline
$0$ & $0$ & $1$ & $1$ & $1 -\nicefrac{1}{2} \, R_{A} - \nicefrac{1}{2} \, R_{B} $ \\ \hline
$0$ & $1$ & $0$ & $\frac{R_{A} + R_{B} - 1}{R_{A} - R_{B}}$ & $\nicefrac{1}{2}$ \\ \hline
$0$ & $\frac{R_{A} - R_{B}}{R_{A} + R_{B} - 1}$ & $0$ & $1$ & $\nicefrac{1}{2}$ \\ \hline
$1$ & $1$ & $0$ & $1$ & $R_{B}$ \\ \hline
$0$ & $1$ & $0$ & $1$ & $\nicefrac{1}{2}$ \\ \hline
$1$ & $0$ & $0$ & $1$ & $\nicefrac{1}{2} -\nicefrac{1}{2} \, R_{A} + \nicefrac{1}{2} \, R_{B}$ \\ \hline
$0$ & $0$ & $0$ & $1$ & $\nicefrac{1}{2}$ \\ \hline
$1$ & $0$ & $\frac{R_{A} + R_{B} - 1}{R_{B} - R_{A}}$ & $0$ & $\nicefrac{1}{2}$ \\ \hline
$\frac{R_{B} - R_{A}}{R_{A} + R_{B} - 1}$ & $0$ & $1$ & $0$ & $\nicefrac{1}{2}$ \\ \hline
$1$ & $1$ & $1$ & $0$ & $R_{A}$ \\ \hline
$0$ & $1$ & $1$ & $0$ & $\nicefrac{1}{2} + \nicefrac{1}{2} \, R_{A} - \nicefrac{1}{2} \, R_{B} $ \\ \hline
$1$ & $0$ & $1$ & $0$ & $\nicefrac{1}{2}$ \\ \hline
$0$ & $0$ & $1$ & $0$ & $\nicefrac{1}{2}$ \\ \hline
$1$ & $1$ & $0$ & $0$ & $\nicefrac{1}{2} \, R_{A} + \nicefrac{1}{2} \, R_{B}$ \\ \hline
$0$ & $1$ & $0$ & $0$ & $\nicefrac{1}{2}$ \\ \hline
$1$ & $0$ & $0$ & $0$ & $\nicefrac{1}{2}$ \\ \hline
\end{tabular}
\end{center}
\caption{Solutions of the equations given by the KKT conditions in the proof of Lemma~\ref{lem:xorindextechnical}, and the corresponding value of $\Pr[\mathcal{F}\mid  \mathcal{E}_{m_A,j}^A, \mathcal{E}_{m_B,i}^B]$. }
\label{table:KKT}
\end{table}
\renewcommand{\arraystretch}{1}

\newpage

\section{Separations two-round hybrid models}

In this section we establish all remaining separation results depicted on Figure~\ref{fig:lattice}. A summary of these results is given in Table~\ref{tab:summary}.

\begin{table}[h]
\begin{center}
\begin{tabular}{c|c|c|c|c|c|c|c|c|c|c|c|c|}
 & ~\B~ & ~\L~ & ~\C~ & \B\B & \B\L & \B\C & \L\B & \L\L & \L\C & \C\B & \C\L & \C\C \\
 \hline
 \B & - & \ref{theo:BBvsS} & \ref{theo:BBvsS} & - & - & - & - & \ref{theo:BBvsS} & \ref{theo:BBvsS} & - & \ref{theo:BBvsS} & \ref{theo:BBvsS}\\
  \hline
 \L & \ref{theo:LvsS} & - & \ref{theo:LvsS} & \ref{theo:LvsS} & - & \ref{theo:LvsS} & - & - & - & \ref{theo:LvsS} & - & \ref{theo:LvsS}\\
 \hline
 \C & \ref{theo:CvsBB} & - & - & \ref{theo:CvsBB} & - & - & - & - & - & - & - & -\\
 \hline
 \B\B & \ref{theo:BBvsS} &\ref{theo:BBvsS} & \ref{theo:BBvsS} & - & \ref{theo:BBvsS} & \ref{theo:BBvsS} & \ref{theo:BBvsS} & \ref{theo:BBvsS} & \ref{theo:BBvsS} & \ref{theo:BBvsS} & \ref{theo:BBvsS} & \ref{theo:BBvsS}\\
  \hline
 \B\L & \ref{theo:BLvsByL} &  \ref{theo:BLvsByL} &  \ref{theo:BLvsByL} &  \ref{theo:BLvsByL} & - & \ref{theo:LvsS} & - &  \ref{theo:BLvsByL} &  \ref{theo:BLvsByL} & \ref{theo:BCvsCB} & \ref{theo:BLvsByL} & \ref{theo:BLvsByL}\\
  \hline
 \B\C & \ref{theo:CvsBB} & \ref{theo:BBvsS} & \ref{theo:BBvsS} & \ref{theo:CvsBB}  & - & - & - & \ref{theo:BBvsS} & \ref{theo:BBvsS} & \ref{theo:BCvsCB} & \ref{theo:BBvsS} & \ref{theo:BBvsS}\\
  \hline
 \L\B & \ref{theo:LvsS} & \ref{theo:BBvsS} & \ref{theo:BBvsS} & \ref{theo:LvsS} & \ref{theo:CBvsBL} & \ref{theo:CBvsBL} & - & \ref{theo:BBvsS} & \ref{theo:BBvsS} & \ref{theo:LvsS} & \ref{theo:BBvsS} & \ref{theo:BBvsS}\\
  \hline
 \L\L & \ref{theo:LLvsS} & \ref{theo:LLvsS} & \ref{theo:LLvsS} & \ref{theo:LLvsS} & \ref{theo:LLvsS} & \ref{theo:LLvsS} & \ref{theo:LLvsS} & - & \ref{theo:LLvsS} &  \ref{theo:CCvsLB} & \ref{theo:LCvsCL} & \ref{theo:LLvsS}\\
  \hline
 \L\C & \ref{theo:CCvsLB} & \ref{theo:CCvsLB}  & \ref{theo:CCvsLB}  & \ref{theo:LvsS} & \ref{theo:CCvsLB}  & \ref{theo:CCvsLB}  & \ref{theo:CCvsLB}  & - & - & \ref{theo:CCvsLB}  & \ref{theo:LCvsCL} & \ref{theo:LCvsCL}\\
  \hline
 \C\B & \ref{theo:CvsBB} & \ref{theo:BBvsS} & \ref{theo:BBvsS} & \ref{theo:CvsBB}  & \ref{theo:CBvsBL} & \ref{theo:CBvsBL} & - & \ref{theo:BBvsS} & \ref{theo:BBvsS} & - & \ref{theo:BBvsS} & \ref{theo:BBvsS}\\
  \hline
 \C\L & \ref{theo:CCvsLB} & \ref{theo:CCvsLB}  & \ref{theo:CCvsLB}  & \ref{theo:LvsS} & \ref{theo:CCvsLB}  & \ref{theo:CCvsLB}  & \ref{theo:CCvsLB}  & - & \ref{theo:CLvsLC} & \ref{theo:CCvsLB}  & - & \ref{theo:CLvsLC}\\
  \hline
 \C\C & \ref{theo:CCvsLB} & \ref{theo:CCvsLB}  & \ref{theo:CCvsLB}  & \ref{theo:CvsBB} & \ref{theo:CCvsLB} & \ref{theo:CCvsLB}  & \ref{theo:CCvsLB}  & - & - & \ref{theo:CCvsLB}  & - & -\\
 \hline
\end{tabular}
\end{center}
\caption{The number in each cell is a reference to the theorem establishing that the set of languages $\mathbf{S}_{row}$ corresponding to the row of that cell is not included in the set of languages $\mathbf{S}_{col}$ corresponding to the column of that cell, i.e., there is a language $\DL\in \mathbf{S}_{row}\smallsetminus \mathbf{S}_{col}$. The sign ``-'' in a cell indicates that $\mathbf{S}_{row}\subseteq \mathbf{S}_{col}$.}
\label{tab:summary}
\end{table}


Several results in this section are proved by reduction to the communication complexity problem \emph{set disjointness} (DISJ). Given two sets $x, y \subseteq [n]$ (usually represented as indicator vectors  $x, y \in \{0,1\}^n$), the task is to decide whether  $x \cap y = \varnothing$ (or equivalently whether $x_i \wedge y_i = 0$ for all $i \in [n]$). Formally, Alice receives $x$ as input, and Bob recieies $y$. The task is to compute 
\[
\text{DISJ}(x,y) = \begin{cases}
1 & \text{ if } x \cap y = \varnothing \\
0 & \text{ otherwise. }
\end{cases}
\]
The communication complexity of $\text{DISJ}$ is high, as shown below. 

\begin{lemma}[Theorem 6.19 in \cite{rao2020communication}]\label{lem:disj}
For every $\epsilon>0$, any randomized protocol that computes DISJ with error probability $\frac{1}{2}-\varepsilon$ must communicate $\Omega(\varepsilon^2 n)$ bits between the two players. 
\end{lemma}

\begin{theorem}\label{theo:CvsBB}
$\C \smallsetminus \B^* \neq \varnothing$. This result hols even for randomized decision algorithms which may err with probability $\frac{1}{2}-\varepsilon$, for every $\epsilon>0$. 
\end{theorem}

\begin{proof}
Let us consider the distributed language 
\[
\mbox{\tt disjointness-on-clique} = \{ (K_n,\ell)\mid (\ell: V(K_n)\mapsto \{0,1\}^n)\land (\forall  i \in [n], \exists  v, \ell(v)_i = 0)\},
\]
 where $K_n$ is the $n$-node clique, and $\ell(v)_i$ is the $i$th entry of the vector~$\ell(v)$. 
 
 We first show that $\mbox{\tt disjointness-on-clique} \in \C$. Note first, that, In one round of \congest, the nodes can check whether they are in a clique. Indeed, recall that every node knows~$n$, and therefore a node with degree less than $n-1$ rejects. Every node orders all nodes, including itself, according to their IDs, providing every node with a rank. Note that all nodes ranks the nodes the same. During the \congest\/ round, each node~$v$ sends $\ell(v)_i$ to the node with rank~$i$ (which could be itself). After the round of communication, the node $v$ with rank $i$ has the set $\{\ell(w)_i : w\in V(K_n)\}$. This node accepts if there exists $w \in V(K_n)$ such that $\ell(w)_i= 0$, and it rejects otherwise.   
 
Let $k \in \mathbb{N}$. We show that  $\mbox{\tt disjointness-on-clique} \notin  \B^k$.  For establishing a contradiction, let us assume that the there exists a $k$-round \bcc\/ algorithm $\mathcal{A}$ deciding {\tt disjointness-on-clique} with error probability $\frac{1}{2}-\varepsilon$.  We show how to use $\mathcal{A}$ for solving DISJ. Let $x,y\in\{0,1\}^n$ be an instance of DISJ. Alice and Bob consider the $n$-node clique $K_n$, with identifiers from~1 to~$n$. Let $e=\{1,2\}$ be the edge connecting the nodes with ID~1 and the node with ID~2. The two players consider the labeling $\ell$ such that $\ell(1) = x$, $\ell(2) = y$,  and $\ell(v) = (1,1,\dots,1)$ for every node $v$ with $\id(v)\geq 3$.  Note that Alice does not know $\ell(2)$, and Bob does not know $\ell(1)$. 
By construction, we have that $\mathcal{A}$ accepts $(K_n, \ell)$ if and only if $\text(DISJ)(x,y)=1$.  The two players simulate the $k$ \bcc\/ rounds of $\mathcal{A}$ as follows. At each round~$r$, Alice sends to Bob the message $m_{1,r}$ broadcasted by the node with ID~1, and Bob sends to Alice the message $m_{2,r}$ broadcasted by the node with ID~2. 
With this information, Alice and Bob can simulate $\mathcal{A}$, tell each other whether one of the nodes they simulate rejects, and then compute $\text{DISJ}(x,y)$. This protocol for DISJ has communication complexity  $\mathcal{O}(k \log n)$, a contradiction with Lemma~\ref{lem:disj}. 
\end{proof}

Note that the proof of Theorem~\ref{theo:CvsBB} shows that the separation holds even for algorithms performing up to $o(n/\log n)$ \bcc\/ rounds.  
  
 \begin{theorem}\label{theo:LvsS}
 $\L \smallsetminus (\B\B \cup \B\C \cup \C\B  \cup \C\C) \neq \varnothing$. This result hols even for randomized decision algorithms which may err with probability $\frac{1}{2}-\varepsilon$, for every $\epsilon>0$. 
 \end{theorem}
 
 \begin{proof}
Let us consider the following distributed language
\begin{align*}
\mbox{\tt disjointness-on-edge} = & \{(P_{2n},\ell)\mid (n>2)  \land (\ell(u_1)\in\{0,1\}^n) \land (\ell(v_1)\in\{0,1\}^n) \\
& \land (\text{DISJ}(\ell(u_1), \ell(v_1))=1) \land (\forall w \notin\{u_1,v_1\}, \ell(v)=\bot)\}, 
\end{align*}
where $P_{2n}$ is the path $(u_n, \dots, u_1, v_1, \dots v_n)$ with $2n$ nodes. A simple \local\/ algorithm guarantees that ${\mbox{\tt disjointness-on-edge} \in\L}$, that is, every node of degree $>2$ rejects, and $u_1$ and $v_1$ exchange their values and accept if and only if $\mbox{DISJ}(\ell(u_1), \ell(v_1))=1$. 

Let $\textbf{S} = \B\B \cup \B\C \cup \C\B \cup \C\C$.  We now show that $\mbox{\tt disjointness-on-edge}\notin\mathbf{S}$. For establishing a contradiction, let us assume that the there exists a 2-round algorithm $\mathcal{A}$ mixing \congest\/ and \bcc\/ for deciding {\tt disjointness-on-clique} with error probability $\frac{1}{2}-\varepsilon$. We show how to use this algorithm to compute DISJ. Let $(x,y)$ be an instance of  DISJ. Alice and Bob construct the instance $(P,\ell)$ of {\tt disjointness-on-edge} where $\ell(u_1) = x$ and $\ell(v_1)=y$. By construction $\text{DISJ}(x,y)=1$ if and only if $(P, \ell)\in\mbox{\tt disjointness-on-edge}$. Of course, Alice does not know $\ell(v_1)$, and Bob does not know $\ell(u_1)$. All messages communicated in the first round of $\mathcal{A}$ by all nodes different from $u_1$ and $v_1$ do not depend on $(x,y)$, and can thus be simulated by the players without any communication. Alice and Bob generate and exchange the messages that $u_1$ and $v_1$ communicate in the first  round of $\mathcal{A}$. If the first round is a \congest\/ round, note that each of the two nodes may generate two messages, one for each or their two neighbors. For the second rounds, Alice and Bob have all the information sufficient to compute what messages will be generate by the nodes, excepted for nodes $u_1$ and~$v_1$, respectively. So Alice and Bob exchange these messages. Alice accepts if all nodes $u_1, \dots, u_n$ accept, and Bob accept if all node $v_1, \dots, v_n$ accept. Then they exchange their decision. This protocol computes DISJ with error probability $\frac{1}{2}-\varepsilon$. This is a  contradiction with Lemma~\ref{lem:disj} as only $O(\log n)$ bits were exchanged by the two players. 
 \end{proof}


\begin{theorem}\label{theo:LLvsS}
$\L\L \smallsetminus (\B\B \cup \L\B \cup \C\C) \neq \varnothing$ This result hols even for randomized decision algorithms which may err with probability $\frac{1}{2}-\varepsilon$, for every $\epsilon>0$. 
\end{theorem}

\begin{proof}

Let us define the distributed language {\tt disjointness-on-path} of pairs $(P,\ell)$ where $P = u_n, \dots, u_1, v_1, \dots v_n$  is a path of length $2n$ ($n >2$), and $\ell: V(P) \rightarrow \{0,1\}^*$ satisfies that $\ell(w) = \bot$ if $w \neq \{u_2, v_{2}\}$ and $DISJ(\ell(u_2), \ell(v_2))=1$. In words, {\tt disjointness-on-path} is the language of paths that have a yes-instance of disjointness in two nodes at distance $2$.
Trivially {\tt disjointness-on-path} $ \in \L^2$: in a protocol every node except $v_1$ and $u_1$ accept. Nodes $v_1$ and $u_1$ learn the values of $\ell(v_2)$ and $\ell(u_2)$ and accept if and only if $DISJ(\ell(u_2), \ell(v_2))=1$. 

Let \(\textbf{S} = (\B\B \cup \L\B \cup \C\C) \). We now show that {\tt disjointness-on-path} $\notin$ {\bf S}. By contradiction, let us assume that there exists an $\nicefrac{1}{2} - \epsilon$-error algorithm $\mathcal{A}$ in {\bf S} solving  {\tt disjointness-on-path}. We show how to define a two-player, $\nicefrac{1}{2} - \epsilon$-error protocol $\Pi$ for DISJ. Let $(x,y)$ be an instance of DISJ and consider the instance $(P,\ell^*)$ of {\tt disjointness-on-path} where $\ell^*(u_2) = x$ and $\ell^*(v_2)=y$. Clearly $(P, \ell^*) \in$ {\tt disjointness-on-path} if and only if $(x,y)\in$ DISJ .

First, let us suppose that $\mathcal{A}$ is a protocol consisting in two {\bcc} rounds. In this case $\Pi$ consists in two rounds of communication. Initially,  using $x$ Alice simulates the first round of $\mathcal{A}$ in every node of $P$ except $v_2$ obtaining messages $\{ m(w): w\in V(P)\setminus\{v_2\}\}$. Similarly, using $y$ Bob simulates the first round of $\mathcal{A}$ on every node of $P$ except $u_2$, obtaining messages $\{ m(w): w\in V(P)\setminus\{u_2\}\}$. Then, in the first round of $\Pi$, Alice and Bob interchange $m(u_2)$ and $m(v_2)$, in order to obtain the pack of messages $M = \{ m(w): w\in V(P) \}$ that every node receives in the first {\bcc} round of $\mathcal{A}$.  The second round is very similar: using $x$ and $M$ Alice simulates $\mathcal{A}$ on every node except $v_2$, obtaining the pack of messages communicated in the second round of $\mathcal{A}$ except for the message of $v_2$. At the same time using $y$ and $M$ Bob simulates $\mathcal{A}$ on every node except $u_2$, obtaining the pack of messages communicated in the second round of $\mathcal{A}$ except for the message of $u_2$. Then, in the second round of $\Pi$ Alice and Bob interchange the second messages of $u_2$ and $v_2$. Finally, Alice simulates the output of every node except $v_2$ and accept if every node accepts. Bob simulates the output of every node except $u_2$ and accept if every node accepts. We deduce that $\Pi$ is an $\epsilon$-error protocol for {\tt disjointness}. Nevertheless, the number of bits communicated in the execution of $\Pi$ corresponds to the two messages broadcasted by $u_2$ and $v_2$, which is $\cO(\log n)$. This contradicts Lemma~\ref{lem:disj}. We deduce that {\tt disjointness-on-path} does not belong to \B\B.

Now, let us suppose that $\mathcal{A}$ is a protocol consisting in a {\local} round followed by a {\bcc round}. In this case $\Pi$ consists in just one round of communication. Initially,  using $x$ Alice simulates first round of $\mathcal{A}$ obtaining that nodes $u_1, u_2$ and $u_3$ learn the value of $x$, and all other nodes have no information of $x$. Then, Alice simulates $\mathcal{A}$ to generate the messages that $u_1$, $u_2$ and $u_3$ communicate in the {\bcc} round and sends these messages to Bob.  Similarly, Bob simulate the rounds of $\mathcal{A}$ and communicates the messages that $v_1$, $v_2$ and $v_3$ communicate in the {\bcc} round and then sends such messages to Alice. After the communication round of $\Pi$, Alice and Bob generate the information communicated by every vertex of the graph except $u_1, u_2, u_3, v_1, v_2, v_3$. Since these nodes have no information of $x$ or $y$, this simulation can be done without sending any further messages between Alice and Bob. Finally, Alice simulates the output of every node except $v_1, v_2, v_3$ and accept if all accept. Bob simulates the output of every node except $u_1, u_2, u_3$  and accept if all accept.
 We deduce that $\Pi$ is an $\epsilon$-error protocol for {\tt disjointness}. Nevertheless, the number of bits communicated in the execution of $\Pi$ corresponds to the messages broadcasted by $u_1, u_2, u_3, v_1 v_2$ and $v_3$, which is $\cO(\log n)$. This contradicts Lemma~\ref{lem:disj}. We deduce that {\tt disjointness-on-path} does not belong to \L\B.

Finally, let us suppose that $\mathcal{A}$ is a protocol consisting in two {\congest} rounds. In this case $\Pi$ consists in just one round of communication. Observe that all messages communicated in the first round of $\mathcal{A}$ by nodes different than $u_2, v_2$ do not depend on $(x,y)$ and can be simulated by the players without any communication. Then, protocol $\Pi$ consists in Alice and Bob generating and interchanging the messages that $u_1$ and $v_1$ communicate in the second round of $\mathcal{A}$. Then Alice accept if every node $u_1, \dots, u_n$ accepts, and Bob accept if every node $v_1, \dots, v_n$ accept. By the correctness of $\mathcal{A}$, with probability $1-\epsilon$, every node accepts if and only if $(P, \ell^*)$ is a yes-instance of {\tt disjointness-on-path}. We deduce that $\Pi$ is an $\epsilon$-error protocol for {\tt disjointness}. Nevertheless, the number of bits communicated in the execution of $\Pi$ corresponds to the messages interchanged by $v_1$ and $u_1$, which are of size $\cO(\log n)$ in total. This contradicts Lemma~\ref{lem:disj}. We deduce that {\tt disjointness-on-path} does not belong to \C\C.
\end{proof}

\begin{theorem}\label{theo:BBvsS} Let $k>1$. For every set \(\textbf{S} = \prod_{i=1}^p\L^{\alpha_i}\B^{\beta_i}\C^{\gamma_i}\) such that \(\sum_{i=1}^p \beta_{i} = k-1 \), it holds that $\B^{k } \smallsetminus {\bf S} \neq \varnothing$. This result hols even for randomized decision algorithms which may err with probability $\frac{1}{2}-\varepsilon$, for every $\epsilon>0$. 
\end{theorem}

  In order to prove this result, we consider the communication complexity known as {\tt pointer chasing}. Let $S_A \cup S_B = [n]$ be a partition of $[n]$  such that $|S_A| = |S_B|$, and a function $f: S \rightarrow S$ such that $f(S_A) \subseteq S_B$ and $f(S_A) \subseteq S_B$. Call $f_A = f |_{S_A}$ and $f_B = f|_{S_B}$. In problem $k$-{\tt pointer chasing} Alice receives as input $S_A$ together with $f_A$ and Bob receives as inputs $S_B$ and $f_B$. The task is to compute the parity of the number of $1$s in the binary representation of $f^k(0)$ ($k$ compositions of $f$). 
  
Let us denote $C^k(f)$ the communication complexity of function $f$ restricted to $k$-round protocols. In \cite{nisan1991rounds} the following result  is shown.

\begin{proposition}\label{prop:pointerchasing} For every $k \geq 1$, 
\begin{itemize}
\item $C^k($k$-\texttt{pointer chasing}) = \cO(\log n)$, and 
\item $C^{k-1}($k$-\texttt{pointer chasing}) = \cO(n - k\log n)$. This result hols even for randomized decision algorithms which may err with probability $\frac{1}{2}-\varepsilon$, for every $\epsilon>0$. 
\end{itemize}
 \end{proposition}

\begin{proof}
Let $\ell \geq 1$. We define a language that belongs to $\B^k$ but does not belong to $\L^*\B^{k-1}$. Let us consider the language {\tt $k$-pointer-chasing-on-long-path ($k$-PCLP)} as the set of paths of length $n$, where the two endpoints of the path have yes-instance {\tt $k$-pointer-chasing}. 
Formally,

\begin{align*}
\mbox{\tt $k$-PCLP} = \Big \{(G, \ell) :&\;  \big (\ell:V(G)  \rightarrow \{0,1\}^* \cup \{\perp\}\big) ~\land~\big( G \mbox{ is a path with endpoints } u, v \big) \\
& \land \big( \ell(w) = \bot ~\wedge~ (\ell(u),\ell(v)) \in k-{\tt pointer-chasing} \big)\Big \}.
\end{align*}

We have that $k$-{\tt PCLP} belongs to $\B^k$. Indeed, a simple algorithm consists in all vertices communicating its degree in the first round, and, at the same time, the two endpoints of the path communicating the successive evaluations of $f$. Once $f^k(0)$ is computed, every node accepts if (1) all except two nodes have degree $2$, and the two remaining nodes have degree $1$ (hence the graph is a path), and (2) the number of $1$ bits in the binary representation of $f^k(0)$ is $1$. 

We now show that $k$-{\tt PCLP} $\notin \L^*\B^{k-1}$. By contradiction, let us suppose that for some $t\geq 0$ there exists an $\nicefrac{1}{2}- \epsilon$-error $\L^t\B^{k-1}$ algorithm $\mathcal{A}$ for $k$-{\tt PCLP}. Given an instance $(f_A, f_B)$ of $k$-{\tt pointer chasing}, we define an instance $(P_n, \ell^*)$ of $k$-{\tt PCLP} as follows. First, consider on an $2n$-node ($n>t$)  path together endpoints $u$ and $v$. Second, assign $\ell(w) = \bot$ to every node except the endpoints. Third, assign $\ell(u) = f_A$ and $\ell(v)=f_B$. We call $P^u_d$ and $P^v_d$ the set of nodes at distance at most $d$ from $u$ and $v$, respectively.  Observe that $(P_n, \ell^*) \in k$-{\tt PCLP} if and only if $(f_A, f_B) \in k$-{\tt pointer-chasing}.

Now consider the following $\epsilon$-error two-player $k-1$-round algorithm $\Pi$ for $\ell$-{\tt pointer chasing}. Alice and Bob virtually construct the input $(P_n, \ell^*)$. We say that Alice owns the nodes in $P^u_t$ and Bob owns the nodes $P^v_t$. All nodes that are not owned by Alice or Bob are called \emph{remaining nodes}. The nodes simulate the $t$ {\local} rounds of $\mathcal{A}$ on the nodes they own and on the remaining nodes.  Notice that the players can simulate these rounds without any communication, since the information of the endpoints are at distance $2n > 2t$. Then, the players perform $k-1$ rounds of communication. On the $i$-th round, Alice and Bob simulate the $i$-th {\bcc} round of $\mathcal{A}$ on all the nodes they own, generating a packages of messages $M^i_A$ and $M^i_B$, respectively. Then, they communicate $M^i_A$ and $M^i_B$ to each other. Finally, each player simulate the $i$-th  {\bcc} round of $\mathcal{A}$ on the remaining nodes generating a package of messages $M^i_R$ . Observe that these latter messages can be generated as they depend only on the messages sent on the previous rounds, and not on the inputs of $f_A$ and $f_B$. Finally, Alice and Bob have each the packages of messages $M^1, \dots, M^{k-1}$ corresponding to the $k-1$ {\bcc} rounds of $\mathcal{A}$. Using that information the players can simulate the output of all nodes they own, as well as the output of the remaining nodes. The players then accept in $\Pi$ if and only if all the nodes they own and the remaining nodes accept in $\mathcal{A}$. By the correctness of $\mathcal{A}$, we obtain that with probability $\nicefrac{1}{2}-\epsilon$, all the nodes in $\mathcal{A}$ accept if and only if Alice and Bob accept in $\Pi$. We deduce that $\Pi$ is an $k-1$-round, $\epsilon$-error protocol for $k$-{\tt pointer-chasing}. However,  in protocol  Alice and Bob communicate $\cO(\ell t \log n) = \cO(\log n)$ bits, which contradicts Proposition \ref{prop:pointerchasing}. We deduce that $k$-{\tt PCLP} $\notin \L^*\B^{k-1}$.

Finally, notice that from Theorem \ref{theo:LalphaBbeta} and the fact that $\C \subseteq \L$, we have that all problems solvable in \(\textbf{S} = \prod_{i=1}^k\L^{\alpha_i}\B^{\beta_i}\C^{\gamma_i}\) can be solved in by an algorithm in $\L^*\B$.
\end{proof}

\begin{theorem}\label{theo:CLvsLC}
$\C\L \smallsetminus \L\C \neq \varnothing$. This result hols even for randomized decision algorithms which may err with probability $\frac{1}{2}-\varepsilon$, for every $\epsilon>0$. 
\end{theorem}

\begin{proof}
Let us consider the graph $S_n = (V=(V_1,V_2),E)$ in which there exists $u^*v^* \in E$ such that : $u^* \in V_1$, $S^1_n = S_n(V_1) \smallsetminus \{v^* \}$ is a star graph with $n$ leaves rooted in $u^*$, $v^* \in V_2$ and $S^2_n = S_n(V_2) \smallsetminus \{u^* \}$ is a start with $n$ leaves rooted in $v.$ Let us consider the following distributed language $\tt{DISJ-edge-star} = \{S_n=(V,E), \ell:V \mapsto \{0,1\} \times [n],  \ell(u) = (b,i) \iff x_i = b, \ell(v) = (b',i'),   \neg  \bigvee \limits_{i \in [n]} x_i \wedge y_i = 1, x,y \in \{0,1\}^{n} \}$ where $x_i = b \iff \ell(v) = (b,i)$ for some $v \in V_1$ and $y_i = b \iff \ell(u) = (b,i)$ for some $u \in V_2.$ 

First, observe that $\tt{DISJ-edge-star} \in \C\L.$  In fact, a protocol $\pi$ in the hybrid \ \congest \ + \local\ model for $\tt{DISJ-edge-star}$ can be described as: in the first round of communication all the nodes in the leaves of each start send its input. More precisely, each $v \in V_1$ and $u \in V_2$ send $x_i$ and $y_i$ respectively for some $i \in [n].$ Observe that after the first round of communication $u^*$, $v^*$ are able to recover $x$ and $y$ from the messages sent by their neighbors. If the inputs of the leaves are not correct in the sense that each index $i$ given in the input is different, they reject. Then, in the second round of communication, the node $u^*$ sends a message containing $x$ to $v^*$ and $v^*$ sends a message containing $y$ to $v^*$. Finally the nodes in the leaves accept and $u^*,v^*$  compute $\text{DISJ}(x,y) =  \bigvee \limits_{i \in [n]} \neg x_i \wedge y_i $ and accept if and only if $\text{DISJ}(x,y) =1$

Now, we are going to show that $\tt{DISJ-edge-star} \not \in \L\C.$ By contradiction, let us assume that there exists a $\frac{1}{2}-\varepsilon$ protocol $\pi$ in the hybrid \  \local  + \congest \  model for $\tt{DISJ-edge-star}.$ We consider an instance $(x,y)$ of the set disjointness problem $\tt{DISJ}.$ Let $n = |x|=|y|.$ We are going to describe a $\frac{1}{2}-\varepsilon$ protocol $\pi'$ for  $\tt{DISJ}.$. Let us consider the instance of $\tt{DISJ-edge-star}$ $(S_n,\ell)$ in which $\ell$ assigns $x_i$ to each leaf in $S^1_n$ and $y_i$  to each leaf in $S^2_n.$ Observe that $(x,y)$ is a yes instance of $\tt{DISJ}$ if and only if $(S_n,\ell)$ is a yes instance for $\tt{DISJ-edge-star}$. Let us define $S^A_n =  S^1_n$ and $S^B_n =  S^2_n$ i.e. we consider the graph induced by each of the star graphs in $S_n.$ We say that Alice and Bob have $S^A_n$ and $S^B_n$ respectively. Since the roots $u^*$ and $v^*$ of $S^1_n$ and $S^2_n$ respectively have an empty input, Alice and Bob can simulate the \ \local \ round of $\pi.$ Then, Alice and Bob simulate the messages sent by the nodes during the \ \congest \ round of $\pi.$ Observe that, since $v^*$ is not in $S^A_n $  and $u^*$ is not in $S^B_n $, Alice cannot simulate he message $m_{v^*,u^*}$ sent by  $v^*$ to $u^*$ and Bob cannot simulate he message $m_{u^*,v^*}$ sent by $u^*$ to $v^*.$ However, since Alice and Bob can simulate the local round, Alice can simulate  $m_{u^*,v^*}$ and Bob can simulate  $m_{v^*,u^*}.$ Thus, Alice sends a message $m_A$ containing  $m_{u^*,v^*}$ to Bob and Bob sends a message $m_B$ containing $m_{v^*,u^*}$ to Alice. Finally, both players can simulate $\pi$ and thus, they compute $\text{DISJ}(x,y).$ However, the cost of the protocol $\pi'$ is $\mathcal{O}(\log n )$ because the size of $m_A$ and $m_B$ is $\mathcal{O}(\log n )$ which is a contradiction.  
\end{proof}

\begin{theorem}\label{theo:LCvsCL}
$\L\C \smallsetminus \C\L \neq \varnothing$. This result hols even for randomized decision algorithms which may err with probability $\frac{1}{2}-\varepsilon$, for every $\epsilon>0$. 
\end{theorem}

\begin{proof}
Consider the language {\tt  special-disjointness} defined by the pairs $(G, \ell)$ such that: (1) $G$ is defined from a path $P = v_1, v_2, v_3, v_4$, a clique $K_n$, and two more vertices $u_1$, $u_2$. Node  $v_4$ is adjacent to an arbitrary node of the clique, and $v_1$ has two pending vertices $u_1$ and $u_2$. (2) $\ell: V(G) \rightarrow (\{0,1\}^*, [4])$ is a function such that:
\begin{itemize}
\item $\ell(w) = (\bot, \bot) $ for every $w \notin V(K_n)$, 
\item $\ell(v_1) = (\bot, 1)$, $\ell(v_2) = (\bot, 2)$,  $\ell(v_3) = (\bot, 3)$,  
\item $\ell(v_4) = (b, 4)$, with $b \in \{0,1\}$
\item $\ell(u_1) = (\bot, x)$ and $\ell(u_2) = \bot, y)$ with $x,y \in \{0,1\}^n$, and
\item $\forall i \in [n]$ such that $x_i y_i = 0$ if and only if $b=1$. In words, $x,y$ are a yes-instance of disjointness if and only if the input of $b$ is $1$.
\end{itemize}
We first show that {\tt special-disjointness} is in $\L\C$. The protocol has two verification algorithms, that are evaluated in parallel. We say that a node accepts if it accepts in both algorithms. The first algorithm, that we call \emph{topology verification} consists in each node $v$ sending its degree and the second coordinate of $\ell(v)$. Then, 
\begin{itemize}
\item If $v$ has degree $n-1$, then $v$ accepts if  and only if $\ell(v) = (\bot,\bot)$ and it has $n-2$ neighbors of degree $n-1$ and one neighbor of degree $n$.
\item If $v$ has degree $n$, then $v$ accepts   if  and only if  $\ell(v) = (\bot,\bot)$ it has $n-1$ neighbors of degree $n-1$ and one neighbor $w$ of degree $2$ such that $\ell(w)_2 = 4$.
\item If $v$ has degree $2$ and $d(v)_2=4$, then $v$ accepts if and only if one of the neighbors of $v$ has degree $n$ and the other neighbor has degree $2$ and has $3$ in its second component.
\item If $v$ has degree $2$ and $d(v)_2 = \{2,3\}$, then $v$ accepts if and only one neighbor has $d(v)_2-1$ and the other $d(v)_2+1$ in their second components.
\item If $v$ has degree $3$, then $v$ accepts if and only if $d(v)_2 = 1$, one neighbor has $2$ ind its second components, and the other two has $\bot$ in their second components.  
\item If $v$ has degree $1$, then $v$ accepts if and only if $d(v)_2 = \bot$, and its neighbor has $1$ in its second components.   
\end{itemize}
Observe that all nodes accept in the topology verification algorithm if and only if $G$ satisfies the properties of the language. Clearly the topology verification algorithm belongs to $\C\L \cap \L\C$. In the following let us assume that every node accepts in the topology verification algorithm. 
The second verification algorithm, called {\it input verification} is used to verify the conditions on $\ell$, specially the last condition. In the input verification algorithm, every node with a degree different than $1$, $2$ or $3$ immediately accepts. For the other nodes $v$, the algorithm is the following:
\begin{itemize}
\item If $v$ has degree $1$ (i.e. $v = u_1$ or $v= u_2$), then in the first round $v$ communicates $\ell_1(v)$ to his neighbor, then accepts.
\item If $v$ has degree $3$ (i.e. $v = v_1$), then in the first round $v$ does not communicate anything. In the second round $v$ receives $x$ and $y$ from two of neighbors. If $x$ and $y$ are not of the same length $n$, then $v$ rejects. Otherwise, it verifies whether $x_i y_i = 0$ for every $i\in [n]$. If the answer is affirmative, it communicates a bit $1$ to the other neighbor, and otherwise it communicates a bit $0$.
\item If $v$ is such that $d(v)=4$ (i.e. $v=v_4$), then $v$ communicates $b = \ell_2(v)$ to its neighbors in the first round and then accept.
\item If $v$ is such that $d(v)=3$ (i.e. $v=v_3$), then $v$ sends nothing in the first round, and receives $b$ from one neighbor. In the second round, he communicates $b$  to the other neighbor and accept. 
\item If $v$ is such that $d(v)=2$ (i.e. $v=v_2$), then $v$ sends nothing in the first two rounds, but receives two bits in the second round from two different neighbors. Then $v$ accept if the two bits are equal.
\end{itemize}
In simple words, the input verification algorithm consists in communicating $x$ and $y$ to $v_1$, then $v_1$ checks whether $x$ and $y$ are disjoint and communicates the answer to $v_2$. At the same time, the bit $b$ is communicated to $v_2$ in two communication rounds. The first round can be done in {\L} as we do not have bandwidth restrictions. The second can be done in {\C} as at most one bit is communicated per edge. We deduce  {\tt  special-disjointness} $\in$ {\L\C}.

We now show that  {\tt  special-disjointness} $\notin$ {\C\L}. Let $\mathcal{A}$ be an $\nicefrac12-\epsilon$-error {\C\L} algorithm for {\tt  special-disjointness}. We show that $\mathcal{A}$ can be transformed into a two-player protocol $\Pi$ for {\tt disjointness}. Observe first that given a yes-instance of {\tt  special-disjointness}, we have that if we change $b$ for $1-b$ on the input of node $v_4$, we obtain a No-instance. However, the nodes in a distance greater than $2$ from $v_4$ can not see this difference. Therefore, as $\epsilon < 1/2$, we have that all vertices in a distance greater than $2$ (in particular $u_1$, $u_2$ and $v_1$) from $v_4$ necessarily accept in $\mathcal{A}$ independently on the values of $x$ and $y$. Following a similar reasoning, we obtain that all nodes at distance greater than $2$ from $u_1$ and $u_2$ (in particular $v_3, v_4$ and all the nodes in the clique) must accept independently of the value of $b$. Therefore, the only node that rejects in the illegal inputs is $v_2$.

Now let $(x,y)$ be an instance of DISJ. In protocol $\Pi$, Alice and Bob assume that they play the role of $u_1$ and $u_2$ in an instance of {\tt  special-disjointness} where the identifiers are chosen arbitrarily, and where $\ell_2(u_1) = x$ and $\ell_2(u_2) = y$ and $\ell_2(v_4)=1$.  Alice simulates the {\congest} round of $\mathcal{A}$ for node $u_1$, generating a message $m_A$. Similarly, Bob simulates the {\congest} round of $\mathcal{A}$ for node $u_2$, generating the message $m_B$. Then Alice and Bob interchange $m_A$ and $m_B$. Using that information, Alice and Bob can simulate $\mathcal{A}$ the message that node $v_1$ sends to $v_2$ in the first and second round. Notice that the message that $v_1$ sends to $v_2$ in the first round does not depend on $x$ and $y$, while the message sent by $v_1$ in the second round only depends on $m_A$ and $m_B$. On the other hand, Alice and Bob can simulate the two rounds of $v_3$ and $v_4$, as their messages do not depend on $x$ and $y$. Then, Alice and Bob obtain the two messages received by $v_2$ from his neighbors $v_1$ and $v_3$. Using that information, the nodes can simulate the output of $v_2$ in $\mathcal{A}$ and accept if and only if $v_2$ accepts. From the correctness of $\mathcal{A}$ and the discussion of previous paragraph, we deduce that $\Pi$ an $\epsilon$-error protocol for DISJ.  Nevertheless, in $\Pi$ only $\cO(\log n)$ bits were communicated in total, which is a contradiction with Lemma \ref{lem:disj}. We deduce that {\tt  special-disjointness}~$\notin$~{\C\L}.
\end{proof}

\begin{theorem}\label{theo:CCvsLB}
 $\C\C \smallsetminus \L\B^* \not = \varnothing$. This result hols even for randomized decision algorithms which may err with probability $\frac{1}{2}-\varepsilon$, for every $\epsilon>0$. 
\end{theorem}
\begin{proof}
Let $G_n = K_{n,n,n,n} = (V = V_1\cup V_2 \cup V_3 \cup V_4,E)$ be a complete $4$-partite graph. Given $X, Y \in \{0,1\}^{n^2}$ we define $\text{DISJ}(X,Y) = \neg \bigvee \limits_{i, j \in [n]} x[i][j] \wedge y[i][j].$ Let us define the distributed language $\tt{DISJ-4-partite-graph} = \{(G_n, \ell_{X,Y}), \ell_{X.Y}:V\mapsto\{0,1\}^*\cup \{\epsilon\}, \ell(i) = x[i] \in \{0,1\}^n \text{ for each } i\in V_1,  \ell(j) = y[j] \in \{0,1\}^n \text{ for each } j\in V_4, \ell(v) = \epsilon \text{ for each } v \in V_2 \cup V_3, \text{DISJ}(X,Y) = 1 \}$ where $X = (x[i][j])_{i,j \in [n]}, Y =(y[i][j])_{i,j \in [n]}.$ 

First, we are going to show that $\tt{DISJ-4-partite-graph} \in \textbf{CC}$. In particular, we are going to show that there exists a two round  \congest  \ protocol $\pi$ for  $\text{DISJ}(X,Y)$. Let $n \in \mathbb{N}$ and $\ell_{X,Y}$ be a labeling for $G_n.$ In the first round of communication, each node $i \in V_1$ sends a message containing $x[i][j]$ to its $j$th-neighbor. Besides,  each node $j \in V_4$ sends a message containing $y[j][i]$ to its $i$th-neighbor. After this round of communication, each node sends $\mathcal{O}(\log n)$ bits to each of its neighbors. Now, each $j \in V_2$ knows the string $(x[i][j])_{i \in [n]}$ and each node $i \in V_3$ knows  the string $(y[j][i])_{j \in [n]}.$ Finally, in the second round of communcation each node $ j \in V_2$ sends   $x[i][j]$ to its $i$th-neighbor and each node $i \in V_3$ sends $y[j][i]$ to its $j$th-neighbor. Observe that each node sends again $\mathcal{O}(\log n)$ bits to each of its neighbors. In addition, each node $j \in V_2$ knows $y[j][i]$ for each $i \in [n]$ and also $x[i][j]$ for each $i \in [n].$ Complementarily, each node $i \in V_3$ knows  $x[i][j]$ for each $j \in [n]$ and $y[j][i])$ for each $j \in [n].$ Now each node $j \in V_2$ computes  $x[i][j]\wedge y[j][i] $ for each $i \in [n]$ and each node $i \in V_3$ computes $x[i][j]\wedge y[j][i]$ for each $j \in [n].$ If $x[i][j]\wedge y[j][i] =1$ for some $i,j \in [n]$ then, the node rejects and otherwise it accepts. Observe that $\pi$ is a two round \congest protocol and  $\pi$ accepts if and only if \text{DISJ}(X,Y) = 1. Thus, $\pi$ decides $\tt{DISJ-4-partite-graph}$ and $\tt{DISJ-4-partite-graph} \in \C\C$ holds. 

Let $k \in \mathbb{N}.$ We show that $\tt{DISJ-4-partite-graph}~\notin~\L\B^k.$ Fix $n \in \mathbb{N}$. By contradiction, let us assume that there exists a (k +1) round $\frac{1}{2} - \varepsilon$-error protocol $\pi$ in \local \ + \ \bcc \ that decides $\tt{DISJ-4-partite-graph}.$ We are going to show that there exists a  $\frac{1}{2} - \varepsilon$-error protocol $\pi'$ for $\tt{DISJ}(X,Y)$ with total communication $\mathcal{O}(n\log n)$. In fact, given an instance $(X,Y)$ of  $\tt{DISJ}$ such that $X,Y \in \{0,1\}^{n^2},$ we define a labeling $\ell_{X,Y}$ for $G_n.$ Observe that $\pi$ accepts in $(X,Y)$ if and only if $\text{DISJ}(X,Y) = 1.$ We define the  complete bipartite graphs $G^A_n =  G_n[V_1,V_2]$ and $G^B_n =  G_n[V_3,V_4]$ corresponding to the first two partitions of $G_n$ and the second two partitions of $G_n$ respectively. Now we say that Alice and Bob have each one $G_A$ and $G_B$ respectively. Let $M^L = (m^L_v)_{v \in V}$ the messages sent by each node $v \in V$ in the first round of $\pi.$ Since the nodes in $V_1$ and $V_4$ are the ones that hold the input and the rest of the nodes have an empty input, Alice and Bob can simulate the \ \local \ round and compute $M^L.$ Then, let $M^{\text{Broad}} = (m^{\text{Broad},s}_v)_{v \in V,s \in [k]}$ be de messages sent by each node in each \ \bcc \ round of $\pi.$ Since Alice and Bob have each one half of the original graph $G_n$ they can simulate the messages $M_A =  (m^{\text{Broad},s}_v)_{v \in V(G^A), s \in [k]}$ and $M_B =  (m^{\text{Broad},s}_v)_{v \in V(G^B),s \in [k]}$ respectively. Now, Alice sends $M_A$ to Bob and Bob sends $M_B$ to Alice. Observe that the size of each message is $\mathcal{O}(2 kn \log n).$ Finally, since Alice and Bob have all the messages exchanged by the nodes in $\pi,$ they can simulate the protocol and compute $\text{DISJ}(X,Y).$ However,  the total cost of the protocol $\pi'$  is $\mathcal{O}(2kn \log n)$ which is a contradiction with Lemma \ref{lem:disj}.
\end{proof}

\end{document}